\documentclass[11pt]{article}
	
\usepackage[letterpaper,margin=0.90in]{geometry}
\usepackage{amsmath, amssymb, amsthm, amsfonts,bm}
\usepackage[inline]{enumitem}
\usepackage{csquotes}
\usepackage{hhline}
 
\usepackage{soul}
\usepackage{cite}
\usepackage{framed}
\usepackage[framemethod=tikz]{mdframed}
\usepackage{appendix}
\usepackage{graphicx}
\usepackage{color}
\usepackage{wrapfig}
\usepackage{multirow}
\usepackage{algorithm, subcaption}
\usepackage[noend]{algpseudocode}
\usepackage[textsize=tiny]{todonotes}
\usepackage{enumitem}
\setitemize{noitemsep,topsep=3pt,parsep=3pt,partopsep=3pt}

\usepackage{xspace}
\definecolor{darkgreen}{rgb}{0,0.5,0}
\definecolor{darkblue}{rgb}{0,0,0.8}
\usepackage{hyperref}
\hypersetup{
    unicode=false,          % non-Latin characters in Acrobat’s bookmarks
    colorlinks=true,        % false: boxed links; true: colored links
    linkcolor=darkblue,          % color of internal links (change box color with linkbordercolor)
    citecolor=darkgreen,        % color of links to bibliography
    filecolor=magenta,      % color of file links
    urlcolor=cyan           % color of external links
}
\RequirePackage[]{silence}
\usepackage[nameinlink,capitalize]{cleveref}

\theoremstyle{theorem}
\newtheorem{theorem}{Theorem}[section]
\theoremstyle{lemma}
\newtheorem{lemma}[theorem]{Lemma}
\theoremstyle{corollary}
\newtheorem{corollary}[theorem]{Corollary}
\theoremstyle{claim}
\newtheorem{claim}[theorem]{Claim}

\theoremstyle{definition}

\theoremstyle{remark}

\newcommand{\ignore}[1]{}

\algnewcommand\algorithmicswitch{\textbf{switch}}
\algnewcommand\algorithmiccase{\textbf{case}}

\algdef{SE}[SWITCH]{Switch}{EndSwitch}[1]{\algorithmicswitch\ #1\ \algorithmicdo}{\algorithmicend\ \algorithmicswitch}%
\algdef{SE}[CASE]{Case}{EndCase}[1]{\algorithmiccase\ #1}{\algorithmicend\ \algorithmiccase}%
\algtext*{EndSwitch}%
\algtext*{EndCase}%
\newcommand{\local}{$\mathsf{LOCAL}$\xspace}

\newcommand{\poly}{\operatorname{\text{{\rm poly}}}}
\newcommand{\polylog}{\operatorname{\text{{\rm polylog}}}}
\newcommand{\floor}[1]{\lfloor #1 \rfloor}

\newcommand{\dist}{\operatorname{dist}}
\newcommand{\outdeg}{\operatorname{outdeg}}

\newcommand{\shift}{\operatorname{\text{\tt shift}}}
\newcommand{\augment}{\operatorname{\text{\tt augment}}}

\newcommand{\merge}{\operatorname{\text{\tt merge}}}

\DeclareMathOperator{\E}{\mathbb{E}}

\renewcommand{\paragraph}[1]{\vspace{0.15cm}\noindent {\bf #1}:}

\newcommand{\FullOrShort}{short}

\ifthenelse{\equal{\FullOrShort}{full}}{
	
  \newcommand{\fullOnly}[1]{#1}
  \newcommand{\shortOnly}[1]{}

  }{

    \newcommand{\shortOnly}[1]{#1}
		\newcommand{\fullOnly}[1]{}
    
  }

\begin{document}

\date{}

\title{Towards the Locality of Vizing's Theorem} 

\author{
Hsin-Hao Su\\
	\small Boston College \\
	\small suhx@bc.edu 
\and
Hoa T.~Vu\\
	\small Boston College \\
	\small vuhd@bc.edu
}

\maketitle

%!TEX root = main.tex

\begin{abstract} 
Vizing \cite{Vizing64} showed that it suffices to  color the edges of a simple graph using $\Delta + 1$ colors, where $\Delta$ is the maximum degree of the graph. However, up to this date, no efficient distributed edge-coloring algorithm is known for obtaining such coloring, even for constant degree graphs. The current algorithms that get closest to this number of colors are the randomized $(\Delta + \tilde{\Theta}(\sqrt{\Delta}))$-edge-coloring algorithm that runs in $\polylog(n)$ rounds by \cite{ChangHLPU18} and the deterministic $(\Delta + \polylog(n))$-edge-coloring algorithm that runs in $\poly(\Delta, \log n)$ rounds by \cite{GKMU18}.

We present two distributed edge-coloring algorithms that run in $\poly(\Delta,\log n)$ rounds. The first algorithm, with randomization, uses only $\Delta+2$ colors. The second algorithm is a deterministic algorithm that uses $\Delta+ O(\log n/ \log \log n)$ colors. Our approach is to reduce the distributed edge-coloring problem into an online and restricted version of balls-into-bins problem. If $\ell$ is the maximum load of the bins, our algorithm uses $\Delta + 2\ell - 1$ colors.  We show how to achieve $\ell = 1$ with randomization and $\ell = O(\log n / \log \log n)$ without randomization.    
\end{abstract}
%\listoftodos
  \thispagestyle{empty}

%\newpage
%\tableofcontents
\newpage
\setcounter{page}{1}
\section{Introduction}

\paragraph{Problem description} Given a simple graph $G = (V,E)$, a $k$-edge-coloring is a mapping $\phi: E \to \{ 1, \ldots, k\}$ that maps each edge to a color where no two adjacent edges are mapped to the same color.

We study the edge-coloring problem in the distributed \textsf{LOCAL} model. In this model, vertices host processors and operate in synchronized rounds. In each round, each vertex sends a message of arbitrary size to each of its neighbors, receives messages from its neighbors, and performs local computations. The time complexity of an algorithm is defined to be the number of rounds used. In the end, each vertex produces its own answer. In the edge-coloring problem, the coloring of an edge $uv$ can either be computed by vertex $u$ or vertex $v$. 

The \textsf{LOCAL} model aims to investigate the {\it locality} of a problem. An $r$-round local algorithm in the \textsf{LOCAL} model implies that each vertex only uses information in its $r$-neighborhood to compute the answer, and vice versa. Therefore, a faster algorithm in the \textsf{LOCAL} model would mean that each vertex uses less local information to compute the answer.
 
Computing edge-coloring in the distributed setting has applications in scheduling problems of wireless networks \cite{ Ramanathan99, GDP08}. It is usually the case that the quality of the solution in these applications depends on the number of colors being used. Therefore, we hope to minimize the number of colors while keeping the coloring to be locally computable.

Let $\Delta$ denote the maximum degree of $G$. It takes at least $\Delta$ colors to color the edges of $G$, since the incident edges to a vertex must be colored with distinct colors. Vizing showed that the edges in every simple graph $G$ can be colored with $\Delta + 1$ colors \cite{Vizing64}. The best sequential algorithm for obtaining such a coloring runs in $O(\min(\Delta m \log n, m \sqrt{n \log n}))$ time \cite{Arjomandi82, GNKLT85}. However, in order to obtain a $(\Delta+1)$-edge-coloring in the \textsf{LOCAL} model, the only known method up to this date is a $O(\mathsf{diameter(G)})$ solution -- a leader collects the topology of the whole graph and then computes the answer.

\begin{displayquote}
``How close can we get to this [Vizing's edge-coloring], while remaining in polylogarithmic-time?''
\end{displayquote}

This is an open problem raised in \cite{FGK17}.  The followings are  the thresholds on the number of colors that have been encountered or tackled by existing algorithms.

\paragraph{$2\Delta - 1$ Threshold}
$(2\Delta - 1)$ is a natural threshold to be investigated because $(2\Delta - 1)$ is the minimum number of colors that can be obtained by the greedy algorithm. Panconesi and Rizzi \cite{panconesi-rizzi} gave a deterministic $(2\Delta - 1)$-edge-coloring algorithm that runs in $O(\Delta + \log^{*} n)$ rounds.

The $(2\Delta - 1)$-edge coloring problem translates to the $(\hat{\Delta} + 1)$-vertex-coloring problem on its line graph, $L(G)$, where $\hat{\Delta} = 2\Delta - 2$ is an upper bound on the maximum degree of $L(G)$. There have been extensive studies on the $(\Delta+1)$-vertex-coloring problem \cite{lub86,alon86, Joh99,GP87,GPS88,linial92,panc92, Kuhn2006On,BEK09,bar15a, BEPS16,FraigniaudHK16,HSS18,ChangLP18,BEG18}. All the results can be applied to obtain a $(2\Delta - 1)$-edge-coloring.

Getting a $\polylog(n)$-rounds {\it deterministic} algorithm for $(2\Delta - 1)$-edge-coloring had been a major open problem. After the progress made in \cite{GS17,GHKMSU17} for $(2+o(1))\Delta$-edge-coloring, the problem was settled recently by \cite{FGK17} and later improved in \cite{GHK18}.

\paragraph{$\Delta + \tilde{\Theta}(\sqrt{\Delta})$ Threshold}
Based on a randomized approach, Panconesi and Srinivasan \cite{panconesi1997randomized} first gave an algorithm that uses $1.6\Delta + O(\log^{1+\delta} n)$ colors and runs in $O(\log n)$ rounds w.h.p. \footnote{W.h.p.~denotes with high probability, which means with probability at least $1 - 1/n^c$ for some constant $c$.}, where $\delta > 0$ is a constant. Later, Dubhashi, Grable, and  Panconesi \cite{dubhashi1998near} gave an algorithm that uses $(1+\epsilon)\Delta$ colors and runs in $O(\log n)$ rounds w.h.p.~ for constant $\epsilon$ and $\Delta > \log^{1+\delta} n$. Later, Elkin, Pettie, and Su \cite{EPS15} removed the requirement that $\Delta \geq \log^{1+\delta} n$ and showed that for constant $\epsilon$, there exists $\Delta_\epsilon$ such that for all $\Delta \geq \Delta_\epsilon$, a $(1+\epsilon)\Delta$-edge-coloring can be obtained w.h.p.~by solving $O(\log^{*} \Delta)$  Lovasz Local Lemma (LLL) instances in the distributed setting. An LLL instance can be computed in $\polylog(n)$ rounds in the distributed setting  \cite{chung2014LLL}.  This  line of research culminated in the work of Chang et al.~\cite{ChangHLPU18}, who showed that for {\it any} $\epsilon = \tilde{\Omega}(1/\sqrt{\Delta})$ (not necessarily a constant), it is possible to obtain a $(1+\epsilon)\Delta$-edge-coloring w.h.p.~by $O(\log(1/\epsilon))$ invocations of distributed LLL with an additive $\poly(\log \log n)$ rounds.

The algorithm of \cite{ChangHLPU18} allows one to find a $\Delta + \tilde{O}(\sqrt{\Delta})$-edge-coloring in $\polylog(n)$ rounds. However, this number of colors seems to be the limit of such randomized approach. An intuitive reason is because in these coloring algorithms, each edge is colored with a (almost) random color out of the $(1+\epsilon)\Delta$ colors. Once an edge is assigned a color, it is permanently colored. Consider a stage of the algorithm where an edge $uv$ is still uncolored but all the incident edges are colored. Let $M(x)$  denote the colors that are not assigned to any incident edges of $x$. The size of $M(u)$ and $M(v)$ are about $\epsilon \Delta$. Note that $M(u) \cap M(v)$ must be non-empty in order to color $uv$. The  $\epsilon\Delta$ colors in $M(u)$ and $M(v)$  are likely to be ``randomly'' sampled out of the $(1+\epsilon) \Delta$ colors. Therefore, the expected number of colors in $M(u) \cap M(v)$ is $O(\epsilon^2 \Delta)$. This implies $\epsilon$ has to be at least $\Omega(1/\sqrt{\Delta})$. 

Very recently, in the dynamic setting, Duan et al. \cite{DuanHZ19} gave a randomized $(1 +\epsilon) \Delta$-edge-coloring algorithm with $O(\poly(1/\epsilon,\log n))$ amortized update time when $\epsilon = \Omega(\sqrt{\Delta} / \log n)$. Incidentally, the number of colors they are able to obtain is also capped at the $\Delta + \Theta(\sqrt{\Delta})$ threshold.

\paragraph{$\Delta + \poly\log (n)$ Threshold} Recently, Ghaffari et al.~\cite{GKMU18} gave a $\poly((1/\epsilon), \log n)$-round deterministic algorithm that uses $(1+\epsilon)\Delta$ rounds provided that $\frac{\epsilon}{\log (1/\epsilon )} =\Omega(\log  n /\Delta)$. A corollary of the result is that it is possible to obtain $\Delta + O(\log n \cdot \log\left(2 + \frac{\Delta}{\log n} \right))$ colors in $\poly(\Delta, \log n)$ rounds. This breaks the $\Delta + \tilde{\Theta}(\sqrt{\Delta})$ barrier, provided that $\Delta = \omega(\log^{c} n)$ for a large enough $c$. However, when $\Delta$ is small, say in a constant degree graph, it is still unclear what the minimum possible number of colors is to color the graph in $\polylog(n)$ rounds.

\subsection{Our Results}
We show that by using $\poly(\Delta, \log n)$ rounds, the number of colors can be pushed down to $\Delta + 2$, which is merely just one more color than that in Vizing's Theorem.  

\begin{theorem}\label{thm:randomized}
There exists a randomized distributed $(\Delta + 2)$-edge-coloring algorithm that runs in $\poly(\Delta, \log n)$ rounds w.h.p. Furthermore, for any $\epsilon \geq (2/\Delta) $, this can be turned into a randomized distributed $ (1+\epsilon)\Delta$-edge-coloring  algorithm that runs in $\poly((1/\epsilon), \log n)$ rounds w.h.p.
\end{theorem}

Chang et al.~\cite{ChangHLPU18} showed that any algorithm for $(\Delta + c)$-edge-coloring based on ``extending partial colorings by recoloring subgraphs'' needs $\Omega(\frac{\Delta}{c} \log \frac{cn}{\Delta})$ rounds. Our algorithms belongs to such category and so a polynomial dependency on $\Delta$ and $\log n$  is necessary.

Since the edge-coloring problem is locally-checkable, by using the derandomization result from \cite{GHK18}, we can convert the randomized algorithm to a deterministic algorithm with a $2^{O(\sqrt{\log n})}$ factor blow-up.

\begin{corollary}\label{cor:deterministic}
For any $\epsilon \geq (2/\Delta)$, there exists a deterministic distributed $(1+\epsilon)\Delta$-edge-coloring algorithm  that runs in $\poly((1/\epsilon))\cdot 2^{O(\sqrt{\log n})}$ time.
\end{corollary}

For deterministic algorithms, we also show that it is possible to obtain a $\Delta + O(\log n / \log \log n)$ coloring in $O(\Delta^6 \cdot \log^{O(1)} n)$ rounds.

\begin{theorem}\label{thm:deterministic}
For any $\lambda > 1$, there exists a deterministic distributed $\Delta + O(\log n / (\log \lambda + \log \log n))$-edge-coloring algorithm  that  runs in $O(\lambda \cdot \Delta^6 \cdot \log^{O(1)} n)$ rounds. Furthermore, this can be turned into a deterministic distributed $(1+\epsilon) \Delta$-edge-coloring algorithm  that  runs in $O(\lambda \cdot (1/\epsilon)^6 \cdot \log^{O(1)} n)$ rounds, provided $\epsilon \geq  K \log n / (\Delta (\log \lambda +  \log \log n))$ for some constant $K >0$. 
\end{theorem}

Comparing Theorem \ref{thm:deterministic} (with $\lambda = O(1)$) to the deterministic algorithms in \cite{GKMU18}, we are using fewer colors ($\Delta + O(\log n/\log \log n)$ v.s.~$\Delta + O(\log n \cdot \log (2 + \frac{\Delta}{\log n} ))$) and fewer number of rounds ($\tilde{O}(\Delta^6)$ v.s.~$\tilde{O}(\Delta^9)$). For bipartite graphs, we show that it is possible to further improve the number of rounds  to just $\tilde{O}(\Delta^4)$. We summarize what coloring can be obtained in $\poly(\Delta, \log n)$ in Table \ref{tbl:result}.

The second statements of Theorem \ref{thm:randomized} and Theorem \ref{thm:deterministic} can be obtained by applying the degree splitting algorithm of \cite{GHKMSU17} to recursively split the edges into subgraphs whose maximum degrees are upper bounded by ${O}(1/\epsilon)$ (see Lemma \ref{lem:convert}). Then we can apply the algorithm on each subgraph to obtain a $\poly((1/\epsilon), \log n)$ algorithm.

\begin{table}[]
\caption{Minimum number of colors that can be obtained in $\poly(\Delta, \log n)$ rounds. The deterministic result on the fourth row is obtained by setting $\lambda = \Delta^{\gamma}$ for some constant $\gamma > 0$ in Theorem \ref{thm:deterministic}.}\label{tbl:result}
\centering
\begin{tabular}{|c|r|r r|r|}
\hline
&Previous&   &  New & Range of $\Delta$  \\ \hhline{|=|=|==|=|}
\multirow{3}{*}{\textsc{Randomized}} &$\Delta + \tilde{O}(\sqrt{\Delta})$ \hspace{6.7mm} \cite{ChangHLPU18} & \multirow{3}{*}{$\Delta + 2$}  & \multirow{3}{*}{Thm. \ref{thm:randomized}}  & $O(\log^c n)$  \\ 
 \cline{2-2} \cline{5-5} 
 &\multirow{2}{*}{\begin{tabular}[c]{@{}r@{}}$\Delta + O(\log n \cdot \log(2 + \frac{\Delta}{\log n} ))$ \\ \cite{GKMU18} \end{tabular}}  & & &  \multirow{2}{*}{$\Omega(\log^c n)$}  \\ 
 & & & & \\ \hhline{|=|=|==|=|}
\multirow{4}{*}{\textsc{Deterministic}}& \multirow{2}{*}{$1.5\Delta$ \hspace{16.8mm} \cite{GKMU18}}
 & & & $O(\frac{\log n}{\log \log n})$  \\ \cline{3-5}
& & $\Delta + O(\frac{\log n}{\log \log n})$  & Thm.~\ref{thm:deterministic}   & $[\Omega(\frac{\log n}{\log \log n}) ,O(\log n)]$  \\ \cline{2-5}
&\multirow{2}{*}{\begin{tabular}[c]{@{}r@{}}$\Delta + O(\log n \cdot \log(2 + \frac{\Delta}{\log n} ))$ \\ \cite{GKMU18} \end{tabular}} & $\Delta + O(\log_{\Delta} n)$  & Thm.~\ref{thm:deterministic}   & $[\Omega(\log n), 2^{O(\sqrt{\log n})}]$   \\ \cline{3-5}
& & $\Delta + 2$ &   Cor.~\ref{cor:deterministic} & $2^{\Omega(\sqrt{\log n})}$   \\ \hline
\end{tabular}
\end{table}

\subsection{A High-Level Description of Our Method}

To get a $\poly(\Delta, \log n)$-round algorithm, we employ the construction of Vizing's coloring in the distributed setting. In Vizing's Theorem, it was shown that it is possible to align the endpoints of an uncolored edge with the same missing color. Thus, the edge can be colored with the missing color. The alignment can be done by shifting a  fan and augmenting along an alternating path. The edges can be colored by this method one after another in a sequential manner. 

However, in the distributed setting, there are two challenges.  First, we have to color a large fraction of uncolored edges simultaneously.  The difficulties arise when we  color  the edges together since there are dependencies among the coloring processes. Second, the number of rounds needed to augment along an alternating path must be at least proportional to its length. It is unclear if short alternating paths always exist (potentially the length of an alternating path can be $\Omega(n)$).

While the first challenge can be resolved by decomposing the coloring process into different phases and applying symmetry breaking techniques carefully, overcoming the second challenge relies on the fact all alternating paths are short. %we did not manage to show it. 
Instead of showing this is true, we chop up long alternating paths to stop the augmentation from propagating. This can be done by placing a blocking edge on it and augment on the portion of the path only from the starting vertex to the blocking edge. In the end, if each vertex is adjacent to at most $\ell$ blocking edges, then the remaining blocking edges can be colored with additional $2\ell - 1$ colors, because the maximum degree graph induced by the blocking edges is at most $\ell$. Therefore, the problem reduces to the following load-balancing problem.

\paragraph{Balls into Bins v.s.~Paths into Vertices}  Consider the following version of balls-into-bins problem. Suppose there are $n$ bins.  The balls arrive in an online fashion. When a ball $x$ arrives, it is given a size-$T$ subset of the bins, $B_x$. We have to decide which bin in $B_x$ to place $x$ into. Each bin can only be the potential choice of at most $t$ balls.  The question is how to minimize the maximum load of the bins. 
   
 The problem of minimizing the maximum number of incident blocking edges per vertex can be thought as a balls-into-bins problem. In particular, each alternating path of length at least $T$ is a ball and each vertex is a bin. Choosing the blocking edge on the first $T$ edges is equivalent to choosing the bin (thus each ball occupies two bins, but we ignore this fact for now). We will show that each vertex can be passed by at most $t$ alternating paths throughout the algorithm.
   
If $t \leq T/\lambda$, then by placing each ball into a random bin, we can show that w.h.p.~the maximum load is $O(\log n /(\log \lambda + \log \log n))$. Moreover, such randomized algorithm can be derandomized by letting each ball go to the bin with the smallest cost (we will specify the cost function later), breaking ties arbitrarily. This is how we obtain our deterministic algorithm that uses additional $O(\log n / (\log \lambda + \log \log n))$ colors in Theorem \ref{thm:deterministic}. Note that the $\Theta(\log n /(\log \lambda + \log \log n))$ bound is tight for the balls-into-bins problem. Additional constraints must be considered in order to further reduce the load.  %In the ordinary balls-into-bins case where there are $n$ balls and $n$ bins, and each ball can be placed into any bin, the power-of-2 choice technique allows one to reduce the load to $O(\log \log n)$ w.h.p.~\cite{}. However, the $\Theta(\log n /\log \log n)$ bound is tight here. Additional constraints must be considered in order to further reduce the load. 

To achieve maximum load of 1, we modify our algorithm such that we place the ball into one of the {\it non-empty} bins randomly. If there is  always at least one non-empty bin in $B_x$ for every ball $x$, it guarantees a maximum load of one. Recall that each ball is a path in our scenario. There are at most $n \Delta^{T-1}$ different paths of length $T$. If $T$ is large enough ($T = \poly(\Delta, \log n)\cdot t$), we can show that w.h.p.~for {\it every} path, at most a constant fraction of the bins on the path are occupied (throughout the algorithm) by taking a union bound over all possible paths of length $T$. Since there are at least a constant fraction of  empty bins, each ball is always guaranteed to land in an empty bin. %Moreover, such property guarantees that the probability that any particular bin is selected by a given ball is $O(1/T)$. This in turn allows us to show the property inductively. 
In the end of the algorithm, the blocking edges induce a matching. After recoloring the matching with a new color, we obtain a $(\Delta + 2)$-edge-coloring.

%!TEX root = main.tex

\section{Notation and Vizing's Theorem}\label{sec:notation}
Consider a partial edge coloring $\phi: E \to \{\bot, 1,2,\ldots, \Delta+1 , \star \}$ of a graph $G$. A partial coloring is {\it proper} if the subgraph induced by each color $i \in \{1,2,\ldots, \Delta+ 1\}$ is a matching. We say an edge $e$ is uncolored if $\phi(e) = \bot$. The color $\star$ is a special color that denotes the color of blocking edges. Let $M(v) \subseteq \{1, 2 ,\ldots, \Delta + 1\}$ denote the {\it missing colors} of $v$ (i.e.~the colors that are not assigned to the incident edges of $v$). Note that $M(v)$ is always non-empty. Fix $m(v)$ to be any missing color of $v$.

Fix a partial edge-coloring $\phi$. Suppose that $v$ is missing $\alpha$, i.e.~$\alpha \in M(v)$, and there is at least one uncolored edge that is incident to $v$. The {\it fan} $F_v$  with respect to $\phi$ is constructed as follow. Let $e_1 = v x_1$ be an uncolored edge incident to $v$. Our goal is to color the edge $vx_1$. If $m(x_1) \notin M(v)$, then there must be an edge $vx_2$ with the color $m(x_1)$.  In particular,  we construct $vx_i$ as the edge with the color $m(x_{i-1})$. We stop this process if for some $i$, we have that $m(x_i) \in M(v)$ or some previous edge $vx_j$ currently has color $m(x_i)$. 

More formally, given $e_{i} = v x_{i}$, we construct edge $e_{i+1}$ if both of the following two conditions do not occur: (i) $m(x_i) \in M(v)$. (ii) $m(x_i) \in \{m(x_1), \ldots, m(x_{i-1}) \}$. Furthermore, $e_{i+1}$ is constructed as the edge that is incident to $v$ and colored in $m(x_{i})$. We repeatedly construct edges $e_1 \ldots e_k$ for some $k \geq 1$ until it is not possible to do it further. The vertex set of $F_v$ consists of the {\it center} $v$ and the {\it leaves} $x_1, \ldots, x_k$. The edge set of $F_v$ is $\{e_1, \ldots e_k\}$. We say $k$ is the degree of $F_v$, $\deg(F_v)$. A fan $F_v$ is an $\alpha \beta$-fan if $\alpha \in M(v)$ and $m(x_{\deg(F_v)}) = \beta$. We again note that this procedure is well-defined since with  $\Delta+1$ colors in the palette, $m(x_i)$ is always non-empty. 

We emphasize that when we use $F_v$ to refer to a fan, that is a shorthand to say that it has $v$ as the center. The notation $F_v$ alone however does not uniquely define a fan since we may grow multiple fans centering at $v$ throughout the algorithm.

The operation $\shift(F_v, i)$ recolors the edges of $F_v$ as follow. For each $1 \leq j \leq i - 1$, color edge $e_j$ with $m(x_j)$. We leave the color of $e_i$ to be the same. An $\alpha \beta$-alternating path is a path whose edges are colored in $\alpha$ and $\beta$ alternatively and it is {\it maximal} of such paths.  That is, no path whose edges are alternating between $\alpha$ and $\beta$ contains it properly. Given an $\alpha \beta$-alternating path $P$, the operation $\augment(P)$ switches the colors from $\alpha$ to $\beta$ and from $\beta$ to $\alpha$ for each edge on $P$. Note that two $\alpha \beta$-alternating paths cannot possibly intersect with each other.

\subsection{Vizing's Theorem and Fan Repair}

%\begin{figure}
%\centering
%\begin{subfigure}[t]{0.3\textwidth}
%\includegraphics[scale = 0.35]{fan} 
%\caption{$P_v$ ends at $x_j$}\label{fig:1a}
%\end{subfigure}
%\begin{subfigure*}[t]{0.3\textwidth}\
%\includegraphics[scale = 0.35]{fan-2}
%\caption{$\augment(P_v)$}\label{fig:1b}
%\end{subfigure}
%\begin{subfigure*}[t]{0.3\textwidth}
%\includegraphics[scale = 0.35]{fan-3}
%\caption{$\shift(F_v,k)$ and color $vx_k$ with $\beta$}\label{fig:1c}
%\end{subfigure}
%\caption{}
%%\caption{$\shift(F_v,k)$ and color $vx_k$ with $\beta$}
%\end{figure}

Vizing's Theorem shows that the uncolored edge in a fan can be properly colored if one recolors some other edges. This is done by using an $\augment$ operation and a $\shift$ operation. We call such a procedure {\it repairing} the fan. We adopt a version of the proof from \cite{MG92}, where it only uses one alternating path instead of potentially two. Given a fan $F_v$ with edges $v x_1, \ldots, v x_k$. If $m(x_k) \in M(v)$, perform $\shift(F_v, k)$ and color $v x_k$ with $m(x_k)$ then every edge in $F_v$ is now colored. Otherwise, it must be the case that there exists $j$, where $1 \leq j \leq k-2$, such that $m(x_k) = m(x_j)$. Let $P_v$ be the $\alpha \beta$-alternating path that starts at $v$. Note that the first edge of $P_v$ is $e_{j+1} = v x_{j+1}$. There are two cases. 

{\bf Case 1:} $P_v$ does not end at $x_j$. We perform $\shift(F_v, j+1)$ followed by $\augment(P_v)$. After the $\shift(F_v, j+1)$ operation, there is only one conflict -- both $v x_j$ and $v x_{j+1}$ are colored with $\beta$. This is resolved by   
$\augment(P_v)$, which changes the color of $v x_{j+1}$ to $\alpha$. Moreover, $\augment(P_v)$ does not create new conflicts, since its ending vertex must be missing either $\alpha$ or $\beta$ (depending on the color of its last edge) before the $\shift(F_v,j+1)$ operation. The operation $\shift(F_v, j+1)$ does not color any edge with $\alpha$ or $\beta$ except for $x v_{j}$. Since $P_v$  does not end at $x_j$ as assumed, the partial coloring remains proper after augmenting along $P_v$.

{\bf Case 2:} $P_v$ ends at $x_j$ (Figure \ref{fig:1a}), we perform $\augment(P_v)$ first. After the augmentation, now $v$ is missing $\beta$ and $x_{j}$ is missing $\alpha$  (Figure \ref{fig:1b}). We let  $m(x_{j}) = \alpha$ (there can be multiple missing colors, but we choose $m(x_j)$ to be $\alpha$). Under the new partial coloring and $m(\cdot)$-values, $F_v$ is still a fan of $v$ by definition. Moreover, $m(x_k) = \beta \in M(v)$. We then perform $\shift(F_v, k)$ and color $v x_k$ with $\beta$ (Figure \ref{fig:1c}).

\begin{figure*}
\centering
 \subcaptionbox{$P_v$ ends at $x_j$\label{fig:1a}}[.3\linewidth]{%
    \includegraphics[scale = 0.3]{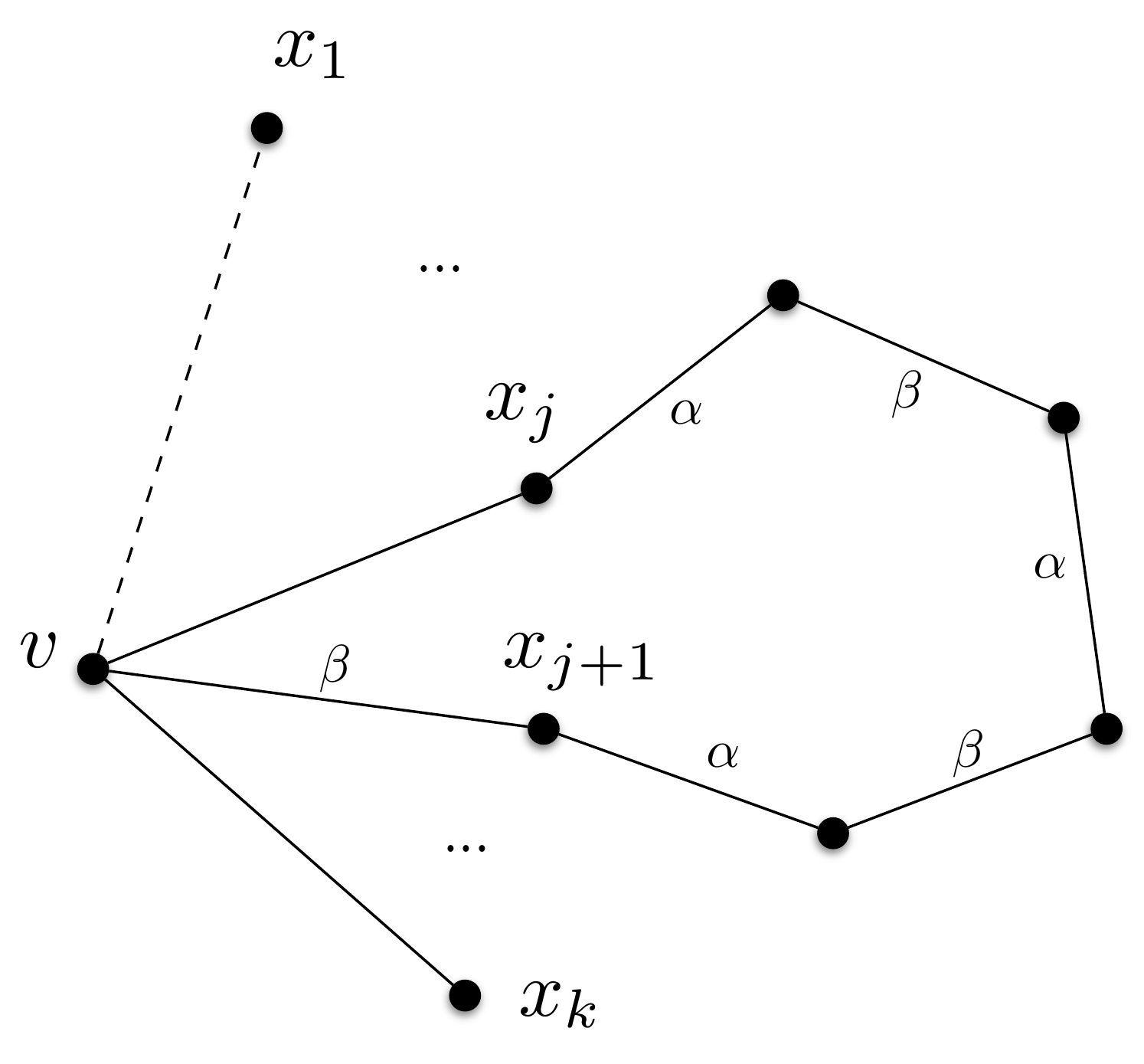}}
  \subcaptionbox{$\augment(P_v)$\label{fig:1b}}[.3\linewidth]{%
    \includegraphics[scale = 0.3]{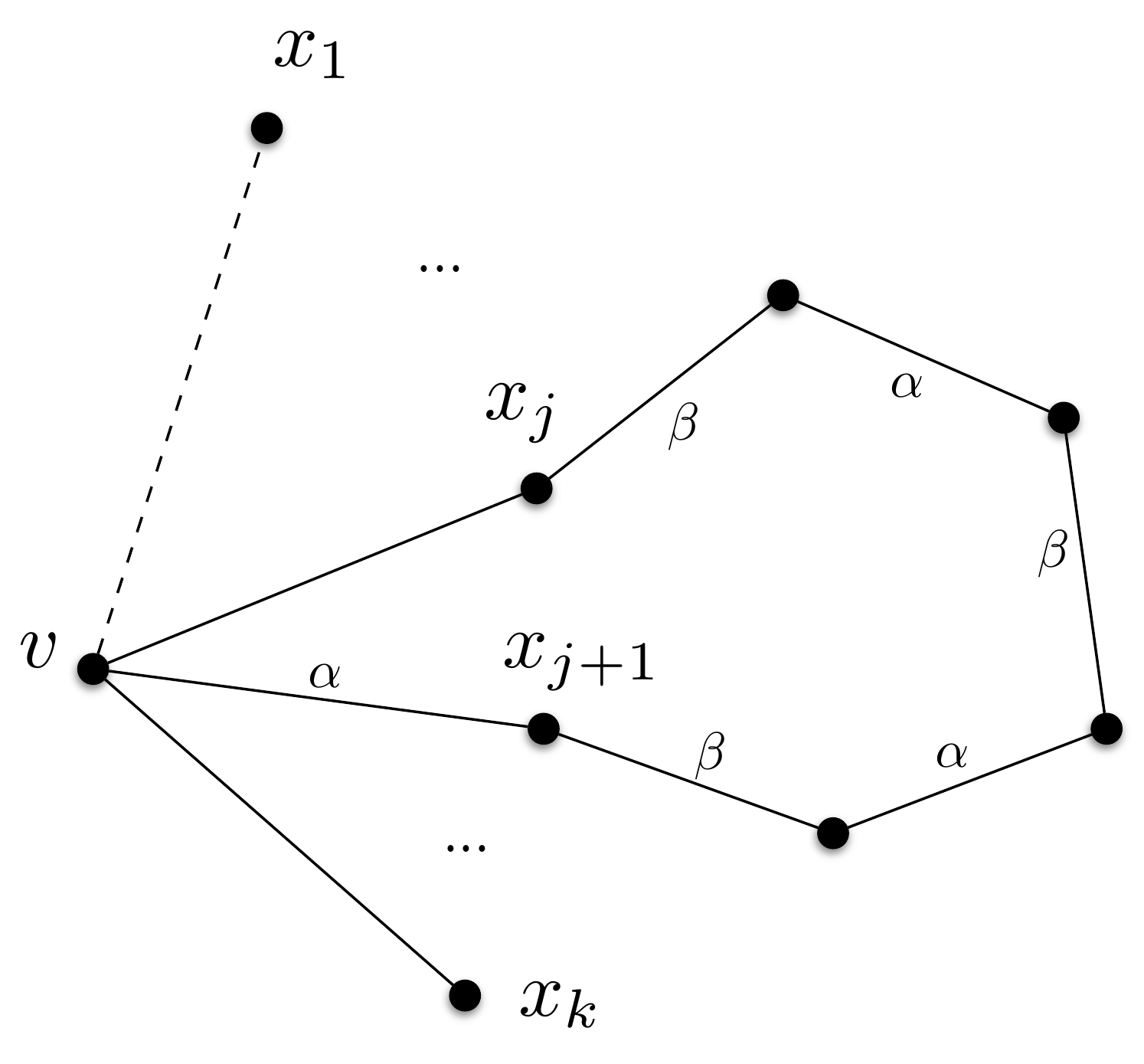}}
      \subcaptionbox{$\shift(F_v,k)$ and color $vx_k$ with $\beta$\label{fig:1c}}[.3\linewidth]{%
    \includegraphics[scale = 0.3]{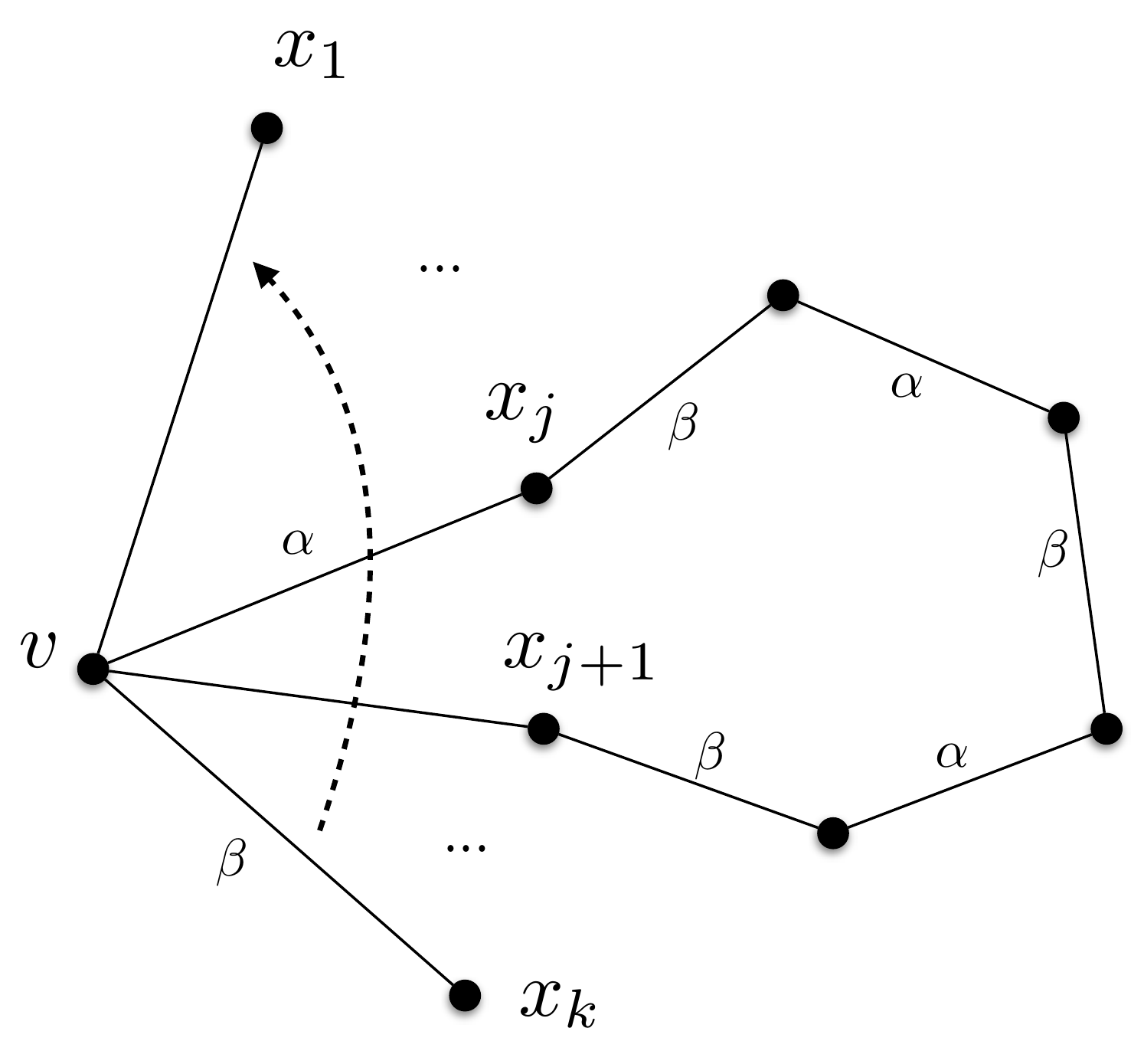}}
\caption{}
\end{figure*}

Therefore, given a fan, we can repair it and thus decrease the number of uncolored edges. Sequentially, if we fix fans one by one, we will be able to color all the edges. To be able to do this efficiently in the distributed setting, we need to address two challenges. 

\begin{itemize}
\item We need to be able to color a large portion of uncolored edges at the same time.

\item  The number of rounds needed to augment along a path is proportional to the length of the path. It is unclear if all alternating paths are short (i.e.~ $\poly(\Delta, \log n)$).

\end{itemize}

%!TEX root = main.tex

\section{Algorithm}
To tackle the first challenge, we divide the fans into different color classes so that the fans in the same color class can be fixed without interfering with each other. The rough idea is to take advantage of the fact that any two $\alpha \beta$-alternating paths do not intersect and the fact that if $u$ and $v$ are in the same color class of a 2-hop coloring then $F_u$ and $F_v$ do not intersect. Therefore, we can repair  all $\alpha\beta$-fans whose centers are in the same color class of a 2-hop coloring together.

%The following property is useful for repairing a set of fans together. 
%\begin{claim}\label{clm:no_intersect} Any two $\alpha \beta$-alternating paths do not intersect. \end{claim}

%Therefore, ideally we may fix all $\alpha \beta$-fans together. However, it is possible the leaves of the fans intersect. By using a 2-hop coloring of the center of the fans, we are able to 

 To overcome the second challenge, the idea is to truncate long alternating paths. Given an alternating path $P = (v_0, v_1, \ldots v_{k} )$, let $P(i)$ denote the truncated version of $P$ defined as follow $$P(i) = \begin{cases} P, &\mbox{if $k \leq i$.} \\ (v_0, v_1, \ldots, v_i), &\mbox{otherwise.} \end{cases}$$
 
 We will truncate alternating paths at length $T = \poly(\Delta ,  \log n)$, where the actual value of $T$ will be specified later. Instead of doing the $\augment$ operation on $P$, we will choose one edge in $P(T)$ to be the blocking edge to block the augmentation propagating down the path. Note that we only need to put a blocking edge on $P(T)$ if $|P| > T$, since if $|P| \leq T$, we can augment $P$ in $O(T)$ rounds. 
 
 \subsection{The Framework}
 
 We describe our algorithm in Algorithm \ref{alg:main}. Let $G^2$ denote the distance-2 graph of $G$ where $(u,v) \in G^2$ if $\dist_G(u,v) \leq 2$. Since the maximum degree of $G^2$ is at most $\Delta + \Delta \cdot(\Delta -1) = \Delta^2$, we can color $G^2$ using $\Delta^2 + 1$ colors. This can be done in $O(\Delta^2 + \log^{*} n)$ rounds \cite{BEK09}. 
 
\begin{algorithm}
\caption{Distributed Fan Repair}\label{alg:main}
\begin{algorithmic}[1]\small
\State Obtain a 2-hop coloring using $\Delta^2 + 1$ colors.		

\While{there exists uncolored edges}
\For{$i = 1, 2, \ldots, \Delta^2 + 1$}\label{ln:outerforloop}		
	\For{$\alpha = $ color $1, 2, \ldots, \Delta + 1$}\label{ln:innerforloop}	
			\State Let $V_{i, \alpha}$ denote color-$i$ vertices that  misses $\alpha$ and have at least one uncolored edges.
			\State Let $\mathcal{F}_{i, \alpha}$ denote the set of fans grown from each vertex $V_{i, \alpha}$. 
			\For{$\beta = $ color $1, 2, \ldots, \Delta + 1$}
				\State Let $\mathcal{F}_{i, \alpha, \beta}$ denote the set of $\alpha \beta$-fans in $\mathcal{F}_{i, \alpha}$.
				\State\label{ln:low_arbor_color}Build a conflict graph $G_{\mathcal{F}_{i, \alpha, \beta}}$ for $\mathcal{F}_{i, \alpha, \beta}$  and color $G_{\mathcal{F}_{i, \alpha, \beta}}$ using $O(1)$ colors.
				\State Let $\mathcal{F}_{i, \alpha, \beta, j}\subseteq \mathcal{F}_{i, \alpha, \beta}$ denote the fans that are  colored in $j$.
				\For{$j =$ color $1,2,\ldots, O(1)$}\label{ln:innermostloop}
					\State \label{ln:last} Repair each fan in $\mathcal{F}_{i, \alpha,  \beta, j}$ that has not been destroyed.
				\EndFor
			\EndFor
		\EndFor
	\EndFor		
\EndWhile
\end{algorithmic}
\end{algorithm}
 The outermost while loop repeats until there is no uncolored edge. We will show in Lemma \ref{lem:outerloop} that  each iteration colors a constant fraction of the uncolored edges. Thus, the while loop uses $O(\log n)$ iterations.  Inside the while loop, we iterate through each color of $G^2$ (Line \ref{ln:outerforloop}). For each color class, we further iterate through all possible values for $\alpha$ (Line \ref{ln:innerforloop}).  Let $V_{i, \alpha}(\phi)$ denote the set of color-$i$ vertices that have at least one incident uncolored edge and miss $\alpha$ w.r.t.~ the partial coloring $\phi$. We often omit the parameter $\phi$ and write it as $V_{i, \alpha}$ when the reference is w.r.t.~the current coloring. %Let $\phi_{i,\alpha}$ denote the coloring during the $i$'th iteration of the outer for loop at Line \ref{ln:outerforloop} and at the beginning of $\alpha$'th iteration of the inner for loop at Line \ref{ln:innerforloop}.  We will show that at least $\Omega(|V_{i, \alpha}(\phi_{i, \alpha})|)$ edges become colored during iteration $\alpha$'th iteration of the inner for loop. 

% During iteration $\alpha$, we first find a maximal matching $M$ of $E_{\alpha}$, the set of edges where both endpoints are missing $\alpha$. Then we color each edge in $M$ with color $\alpha$. Now $V_{\alpha}(\phi)$ must form an independent set. 
 
% Next, we obtain a 2-hop coloring on $V_{\alpha}(\phi)$. This can be done in $O(\Delta^2 + \log^{*} n)$ rounds by \cite{}. Let $V_{\alpha, i}$ denotes the set of vertices in $V_{\alpha}$ that are colored with color $i$. 

  For each $v \in V_{i,\alpha}$, we grow a fan $F_v$ centered at $v$. Let $\mathcal{F}_{i, \alpha}$ be a collection of all such fans. It is guaranteed that the fans do not intersect, since $V_{i,\alpha}$ is an independent set in $G^2$.  Next, we iterate $\beta$ from color 1 to $\Delta + 1$. Let $\mathcal{F}_{ i, \alpha, \beta}$ denote all the $\alpha\beta$-fans in $\mathcal{F}_{i, \alpha}$. Ideally, we want to be able to repair all fans in $\mathcal{F}_{i, \alpha,  \beta}$ simultaneously, since their alternating paths do not intersect. However, it is possible that an alternating path of a fan $F_v$ ends at a node of another fan $F_w \in \mathcal{F}_{i, \alpha,  \beta}$. In this case, we might not be able to repair them together, because the $\augment(P_v)$ step may {\it destroy} the structure of $F_w$ (i.e.~change $M(x)$ for $x \in F_w$ or change the color of the edges in $F_w$).  Note that if $P_v$ crosses a node in $F_w$ and does not end at a node in $F_w$, augmenting along $P_v$ does not affect the structure of $F_w$, since in this case, $P_v$ cannot contain any edge of $F_w$. In other words, if $P_v$ crosses but does not end at a node $x \in F_w$, then   $m(x) \neq \alpha$ and $m(x) \neq \beta$. Therefore,  augmenting $P(v)$ or the prefix of $P_v$ before the blocking edge (if $|P_v| > T$) cannot change $m(x)$.
  
  To resolve this issue, we consider the directed graph $G_{\mathcal{F}_{i ,\alpha, \beta}} = (V_{\mathcal{F}_{ i ,\alpha,\beta}}, E_{\mathcal{F}_{i ,\alpha, \beta}})$, where $V_{\mathcal{F}_{ i , \alpha,\beta}} = \{F_v \mid F_v \in \mathcal{F}_{i, \alpha,\beta} \}$ and $(F_u, F_w) \in E_{\mathcal{F}_{i, \alpha,\beta}}$ if $|P_v| \leq T$ and $P_v$ ends at any nodes of $F_w$. Note that we only consider $|P_v| \leq T$, because if $|P_v| > T$, we will only do modification on $P_v(T)$. Such modification will not affect the fans that intersect with $P_v(T)$.
  
  Since the out-degree of every vertex in $G_{\mathcal{F}_{ i , \alpha,\beta}}$ is at most 1, the arborcity of $G_{\mathcal{F}_{i , \alpha, \beta}}$ is $O(1)$. If we consider $G_{\mathcal{F}_{i,\alpha,\beta}}$ as an undirected graph, we can color $G_{\mathcal{F}_{i , \alpha, \beta}}$ using $O(1)$ colors in $O(\log n)$ rounds \cite{barenboim2010sublogarithmic}, if $G_{\mathcal{F}_{ i , \alpha,\beta}}$ were the underlying communication graph. Since each edge in $G_{\mathcal{F}_{ i , \alpha,\beta}}$ is stretched by a factor of at most $T$ in $G$, we can simulate the coloring algorithm in $O(T \cdot \log n)$ rounds in $G$. Now we iterate through each color $j$ of the coloring and process the fans colored in $j$ together (Line \ref{ln:innermostloop}).  Since the alternating path of each fan colored in $j$ does not end at any other fan colored in $j$, nor do the alternating paths intersect, we can repair them simultaneously. 
  
Therefore, it is possible to repair all undestroyed fans in  $\mathcal{F}_{i,\alpha,\beta,j}$ simultaneously, using a $\shift$ operation and possibly an $\augment$ operation for each fan. For a fan $F_v$ with its alternating path $P_v$ whose first edge is the edge between $v$ and the $(j+1)$'th leave of $F_v$, if $|P_v| \leq T$, we repair the fan as outlined in Section \ref{sec:notation}. Otherwise, we select an edge of $P_v(T)= (v_0,v_1,\ldots,v_T)$ to be the {\it blocking edge}, say $(v_{i-1}, v_{i})$, where $1 \leq i \leq T$. We color $(v_{i-1}, v_i)$ with a special color $\star$ and then perform $\shift(v, j+1)$ followed by $\augment(P_v(i-1))$ (Case 1 of the repairing step). After the repairing step, the uncolored edge in the fan becomes colored and the partial coloring remains proper. We will show that after $O(\log n)$ iterations of the outermost while loop, all the edges become colored (some possibly in color $\star$). 

\begin{lemma}\label{lem:outerloop} In each iteration of the while loop, at least a constant fraction of uncolored edges become colored. Thus, it takes $O(\log n)$ iterations to color all the edges. \end{lemma} 
\begin{proof}
Consider an iteration of the while loop. Let $\phi_0^{-1}(\bot)$ denote the set of uncolored edges at the beginning of the while loop. Let $\phi_{i,\alpha}$ be the partial coloring during the $i$'th iteration of the outer {\bf for loop} at Line \ref{ln:outerforloop} and at the beginning of $\alpha$'th iteration of the inner {\bf for loop} at Line \ref{ln:innerforloop}. Let $\phi_{i}$ denote the partial coloring at the beginning of $i$'th iteration of the outer {\bf for loop} at Line \ref{ln:outerforloop}. 

For each $uv \in \phi_0^{-1}(\bot)$, we orient the edge from $u$ to $v$ if the color of $u$ is smaller than that of $v$; otherwise we orient it from $v$ to $u$. We will show that \[
\sum_{\alpha = 1}^{\Delta+1}|\mathcal{F}_{i,\alpha}| \geq  \sum_{u \in V_i} \outdeg(u) / 2,\] where $V_i$ denotes the set of color-$i$ vertices. %Consider that a vertex $v \in V_{i,\alpha}(\phi_{i}) \setminus  V_{i,\alpha}(\phi_{i,\alpha}) $. 

For each $v \in V_i$, let $C(v)$ denote the first $\outdeg(v)$ missing colors of $v$ w.r.t.~$\phi_i$ so that $|C(v)| = \outdeg(v)$. For each $\alpha \in C(v)$, if $v \in V_{i,\alpha}(\phi_{i,\alpha})$, then a fan centered at $v$ will be created during iteration $\alpha$ (Recall that $V_{i,\alpha}(\phi_{i,\alpha})$ denotes the set of color $i$-vertices that have at least one uncolored incident edge and are missing $\alpha$ with respect to the coloring $\phi_{i,\alpha}$). Otherwise, if $v \notin V_{i,\alpha}(\phi_{i,\alpha})$, it must be caused by the repairing step of some color-$i$ vertex $u$ that changes the color of an incident edge of $v$ to $\alpha$ during iteration $\alpha'$ for some $\alpha' < \alpha$. In other words, there is a fan $F_u$ that was grown in iteration $\alpha' < \alpha$ and repairing $F_u$ causes $v \notin V_{i,\alpha}(\phi_{i,\alpha})$. Repairing a fan $F_u$ can only change the missing color set (i.e.~$M(\cdot)$) of its leaves, its center, and the other endpoint of the alternating path. Since $u$ and $v$ are both in $V_i$, they cannot be neighbors. Therefore, $v$ must be either the endpoint of the alternating path $P_u$ or $v = u$. 

 {\bf Case 1:} $v$ is the end of the alternating path of $F_u$.  It must be the \texttt{augment} operation of $F_u$ that changes the missing color of $v$ during the repairing step of $F_u$. 
 
 {\bf Case 2:} $u = v$. Let $x_1 \ldots x_k$ be the leaves of $F_u$. If an edge incident to $u$ becomes colored in $\alpha$ during iteration $\alpha'$, it must be the case that $\alpha = m(x_k)$ and $\alpha \in M(u)$. We perform $\shift(F_u,k)$ operation and color $u x_k$ in $\alpha$. Note that in this case the $\augment$ operation is not performed.

%it cannot be the case that an uncolored edge  edge $vw$  going  out from $v$ is colored in $\alpha$, because $\phi_{i,\alpha}$ is the coloring before iteration $\alpha$ is executed.  Moreover, it cannot be colored during the repairing step of $w$, because the color of $w$ is larger than that of $v$ so that the repairing step of $w$ will be executed later than that of $v$.

Therefore, if $v \notin V_{i,\alpha}(\phi_{i,\alpha})$, we can blame it to some fan $F_u$. Note that each fan $F_u$ can be blamed at most once, since Case 1 and Case 2 are mutually exclusive. Since for each $v$ and $\alpha \in C(v)$, either it creates a fan or it blames another fan, we have
\begin{align*}
\sum_{u \in V_i} \outdeg(u) =  \sum_{u \in V_i} |C(u)| &\leq \sum_{\alpha=1}^{\Delta + 1}|\mathcal{F}_{i,\alpha}| + \mbox{(\#blames)} \\
&\leq 2 \sum_{\alpha=1}^{\Delta + 1}|\mathcal{F}_{i,\alpha}|  ~.
\end{align*}

The last inequality follows since $\mbox{(\#blames)} \leq \sum_{\alpha=1}^{\Delta + 1}|\mathcal{F}_{i,\alpha}|$.

%Repairing a fan $F_u$ can only decrease the missing color set (i.e.~$M(x)$) of its leaves, its center, and the other endpoint of the alternating path. Since both $v$ and $u$ are colored in $i$, they must not be neighbors. Therefore, $v$ must be the end of the alternating path in order for this to happen. Therefore, if $v \notin V_{i,\alpha}(\phi_{i,\alpha})$ where $\alpha \in C(v)$, it can be blamed to the event that $u \in V_{i,\alpha'}(\phi_{i,\alpha'})$ for some $u$ and $\alpha' \in C(u)$. Moreover, each  event can be blamed  at most once. Therefore, we have

%$$\sum_{\alpha = 1}^{\Delta+1} |V_{i,\alpha}(\phi_{i, \alpha})| \geq \sum_{v \in V_i} \sum_{\alpha \in C(v)} [v \in V_{i,\alpha}(\phi_{i,\alpha})] \geq \sum_{v \in V_i} |C(v)| / 2 = \sum_{v \in V_i} \outdeg(v) / 2 $$

Note that at most half of the fans are destroyed, since repairing a fan can only destroy at most one other fan.  If a fan is not destroyed, then the uncolored edge in the fan must become colored. Therefore, at least $\sum_{i = 1}^{\Delta^2 + 1} \sum_{\alpha = 1}^{\Delta + 1}  |\mathcal{F}_{i,\alpha}|/2$ uncolored edges are colored. The total number of edges that are colored during the while loop is at least
$$\sum_{i = 1}^{\Delta^2 + 1} \sum_{\alpha = 1}^{\Delta + 1} |\mathcal{F}_{i,\alpha}| / 2  \geq  \sum_{i = 1}^{\Delta^2 + 1} \sum_{v \in V_i} \outdeg(v) / 4 = |\phi_0^{-1}(\bot)| /4~. $$

This implies that at least $1/4$ of the uncolored edges become colored during the iteration. Thus, after $O(\log n)$ iterations, all edges are colored.

\end{proof}

\subsection{Load Balancing the Blocking Edges}
By Lemma \ref{lem:outerloop}, Line \ref{ln:last} of Algorithm \ref{alg:main} is executed for at most $t = K \Delta^4 \log n$ times for some constant $K > 0$. Note that in each execution of Line \ref{ln:last}, the alternating path of each fan in $ \mathcal{F}_{i,\alpha, \beta, j}$ does not intersect alternating paths of other fans in $\mathcal{F}_{i, \alpha,\beta,j}$ .

%However, it is possible that a vertex becomes incident to multiple edges colored in $\star$. 

%We will first show that after $O(\log n)$ iterations of the outermost while loop, all the edges become colored (some possibly in color $\star$). 

%It is possible that a vertex becomes incident to multiple edges colored in $\star$ and the edges need to be recolored with distinct colors. 
Let $\ell(v)$ denote the {\it load} of $v$, which is defined to be the  number of incident edges to $v$ colored in $\star$. Also, let $\ell(G) = \max_{v \in V} \ell(v)$. The hope is that at the end of the algorithm, $\ell(G)$ will be small. Since the maximum degree induced by $\star$-edges is $\ell(G)$, we can recolor  $\star$-edges using $2\ell(G)-1$ additional colors using the algorithm of Panconesi and Rizzi \cite{panconesi-rizzi} in $O(\ell(G) + \log^{*} n)$ rounds.

In each execution of Line \ref{ln:last} of Algorithm \ref{alg:main}, each vertex can only increase its $\star$ degree by one. This is because the alternating paths of the fans in $\mathcal{F}_{i,\alpha, \beta,j}$ do not intersect, and thus each vertex can be passed by at most one alternating path. Recall that $t$ is the number of times Line \ref{ln:last} of Algorithm \ref{alg:main} is executed. Therefore, each vertex can be passed by at most $t$ alternating paths throughout the algorithm. 

On the other hand, for each alternating path $P$ of length at least $T$, it has $T$ choices to place the blocking edge on $P(T)$. If we place the blocking edge uniformly at  random on each path, the expected load for each vertex $v$ is, $\E[\ell(v)] \leq \frac{2t}{T} = O(1/\lambda)$, provided $T \geq t \cdot \lambda$. We can show by Chernoff bound that $\ell(G) = O(\log n /(\log \lambda + \log \log n))$ w.h.p. However, to obtain a $(\Delta + 2)$-edge coloring, we need to make $\ell(G) = 1$ so that the $\star$-edges form a matching. In the rest of this section, we show how to achieve $\ell(G) = 1$ using randomization and $\ell(G) = O(\log n /(\log \lambda + \log \log n))$ without randomization.

\paragraph{Achieving $\ell(G) = 1$} Let $T = K \Delta^{7}\cdot  t = K \Delta^{11} \log n$ for some sufficiently large constant $K$. Instead of randomly 
placing the blocking edge, we do the following. Given a path $P$, let $E_0(P)$ be the set of edges whose endpoints have zero loads. In each alternating path $P$, we choose an edge from $E_0(P)$ uniformly at random to be the blocking edge. We will show that at any time during the algorithm, for every path $P$  of length $T$ ($P$ is not necessarily an alternating path), $|E_0(P)| \geq T / 15$ w.h.p. Therefore, given an alternating path of length at least $T$, it is always possible to choose a blocking edge whose both endpoints have zero loads w.hp. If some alternating path of length at least $T$ with $|E_0(P)| < T/15$ occurs during Line $\ref{ln:last}$, we say the algorithm {\it fails} and let our graph and coloring {\it freeze} there. In that case, $P$ will not be processed. If the algorithm did not fail, then we have achieved $\ell(G) = 1$.

%This allows us to conveniently say that no vertex is adjacent to more than one blocking edge at any point in time.

\begin{lemma}After the last execution of Line \ref{ln:last} of Algorithm \ref{alg:main}, it holds w.h.p.~ that for every path $P$ of length $T$, we have that $|E_0(P)| \geq T/15$. \end{lemma}

{%\color{red} 
\begin{proof}
%Let $H(q)$ denote the hypothesis  that after the $q$'th execution of Line \ref{ln:last} of Algorithm \ref{alg:main}, for every path $P$ of length $T$, $|E_0(P)| \geq T/15$ and the algorithm did not fail up to this point. We will show that $H(q)$ holds w.h.p.~inductively for $1 \leq q \leq t$. Note that the induction hypothesis is that $H(q)$ holds w.h.p.

Fix a path $P$ of length $T$. We will show that the probability $\Pr(|E_0(P)| < T/15) \leq 1/(\poly(n) \cdot \Delta^{T})$ after the $t$'th execution of Line \ref{ln:last} of Algorithm \ref{alg:main}. By taking a union bound over all possible paths of length $T$, we will conclude that w.h.p.~for every path $P$ of length $T$, $|E_0(P)| < T/15$. This also implies the algorithm did not fail throughout execution $1 \ldots t$ w.h.p.

We say a vertex $v$ is {\it occupied} if $\ell(v) \geq 1$. Otherwise, it is {\it empty}.  Similarly, we say an edge is occupied if any of its endpoints is occupied. Otherwise, it is 
empty.

We say an alternating path is adjacent to an edge $uv$ in $P$ if  $u$ or $v$ is in the alternating path. We divide all the alternating paths from the beginning up to after the  $t$'th execution of Line \ref{ln:last} of Algorithm \ref{alg:main} that intersect $P$ into two groups:
\begin{align*}
\mathcal{P}_1 & = \{ \text{alt. paths that are adjacent to at least $\frac{T}{1000}$ edges in $P$} \} \\
\mathcal{P}_2 & = \{ \text{alt. paths that are adjacent to fewer than $\frac{T}{1000}$ edges in $P$} \}.
\end{align*}
%and  $\mathcal{P}_2  = \overline{\mathcal{P}_1 }$.

Because the alternating paths in each  execution of Line \ref{ln:last} of Algorithm \ref{alg:main} are vertex-disjoint, for each edge in $P$, at most 2 alternating paths are adjacent to it in each execution. Thus, at most $2t$ alternating paths are adjacent to that edge in total. This implies the number of pairs $(e, \tilde{P}) \in P \times \mathcal{P}_1$ where $\tilde{P}$ is adjacent to $e$ is at most $2tT$. Since each alternating path in $\mathcal{P}_1$ consumes at least $T/1000$ pairs, we have
$$
|\mathcal{P}_1| \leq \frac{2t T}{T/1000} = 2000t~.
$$
Note that we only place at most one blocking edge on an alternating path (i.e., if the path is truncated). Each blocking edge can cause at most three edges in $P$ to be occupied (this corresponds to the case the blocking edge is in $P$). Thus, the alternating paths in $\mathcal{P}_1$ can cause at most $6000t$ edges in $P$ to be occupied.

Now we bound the number of occupied edges in $P$ caused by the blocking edges of alternating paths in $\mathcal{P}_2$.
Let  $P = u_0 \ldots u_{T}$. Let $E'(P) = \{ u_0 u_1 , u_2 u_3, \ldots, u_{2 \cdot (\lceil T/2 \rceil - 1)}u_{2 \cdot (\lceil T/2 \rceil - 1) + 1} \}$ denote the set of every  other edge in $P$ so that $|E'(P)| = \lceil T/2 \rceil$. We say an edge is {\bf adjacent} to $E'(P)$ if it shares an endpoint with an edge in $E'(P)$. 

We say a vertex $u$ is {\bf in} $E'(P)$ if it is an endpoint of an edge in $E'(P)$. We orient every edge (including those not in $P$) $e = uv$ that is adjacent to $E'(P)$  as follows: Without loss of generality, let $u$ be the vertex that is in $E'(P)$. If $v$ is not in $E'(P)$, we orient $e$ from $u$ to $v$. Otherwise, we orient the edge from the vertex that appears earlier in $P$ to the vertex that appears later in $P$.

Given an edge $e=uv \in E'(P)$, we say $e$ is {\bf outwardly occupied} if any of the {\bf outgoing edges} from $u$ or $v$ is a blocking edge of some alternating path in $\mathcal{P}_2$. Next, we will bound the number of outwardly occupied edges in $E'(P)$. 

The reason we are using ``outwardly occupied'' instead of ``occupied'' is because it resolves some dependency issues. In particular, each alternating path in $\mathcal{P}_2$ can only make at most one edge in $E'(P)$ outwardly occupied, while it can make two edges occupied. Moreover, as we will see later, the number of occupied edges in $E'(P)$ is at most 3 times the number of outwardly occupied edges.

%We use the following lemma to upper bound the probability that $e$ is outwardly occupied conditioned on any outcomes of other edges in $E'(P)$. The reason that it can be upper bounded under such conditioning is because each alternating path in $\mathcal{P}_2$ is only affecting a small portion of $P$. 

% \begin{lemma}\label{lem:occupied_prob} Suppose that $H(q-1)$ holds w.h.p. Let $X_e = 1$ if $e$ is outwardly occupied by a blocking edge in an alternating path in $\mathcal{P}_2$ and $X_e = 0$ otherwise.  
%Fix an edge $e\in P$. Let $Z \subseteq E'(P)$ be any subset of edges of $E'(P)$ where $e \notin Z$, the probability that $X_e = 1$ conditioned on any outcomes of $X_e'$ for all $e' \in Z$, $\Pr(X_e = 1\mid \bigwedge_{e' \in Z} X_{e'}  ) < 100t /T$. \end{lemma}
 
Now we hope to upper bound the total number of outwardly occupied edges $X = \sum_{e \in E'(P)} X_e$ where $X_e$ is the indicator variable for the event $e$ is outwardly occupied. Consider an edge $e$, the following is a rough argument to bound $\Pr(X_e = 1)$. The probability that an alternating path makes $e$ outwardly occupied is $O(1/T)$, since it has $\Omega(T)$ empty edges to choose from. Moreover, there are at most $O(t)$ alternating paths that are adjacent to $e$ throughout the algorithm. Thus, by taking a union bound over these paths, we have $\Pr(X_e = 1) = O(t) \cdot O(1/T) = O(t/T)$. As a result, $\E[X] =O(T \cdot (t/T)) = O(t)$. However, it will be difficult to obtain a tail bound on $X$ directly due to dependency issues.

Since we need to apply the strongest form of Chernoff bound in order to obtain a very sharp bound ($1/(\poly(n)\cdot \Delta^{T})$) on the probability that $X$ deviates from $\E[X]$, we couple random variables $\{ X_e \}_{e \in E'(P)}$ with independent Bernoulli random variables $\{Y_e \}_{e \in E'(P)}$ with success probability at most $150t/T$ in the following lemma whose proof is deferred to the end of this proof.
 \begin{lemma} \label{lem:coupling} For every $e \in E'(P)$, let $X_e = 1$ if $e$ is outwardly occupied by a blocking edge in an alternating path in $\mathcal{P}_2$ and $X_e = 0$ otherwise. Let $\{ Y_{e} \}_{e \in E'(P)}$ be independent Bernoulli random variables with success probability at most $150t/T$.  Let $X = \sum_{e \in E'(P)} X_e$ and $Y = \sum_{e \in E'(P)} Y_e$. For any $M \geq 0$, we have $\Pr(X > M )  \leq \Pr(Y > M )$. \end{lemma}
% Let $B(n,p)$ denote  a binomial random variable with $n$ trials and success probability $p$.  Let $Y = B(|E'(P)|,150t/T)$.  We have 
Therefore,
\begin{align*}
\E[Y]   \leq \lceil T/2 \rceil \cdot \frac{150t}{T}  
< 100t~.
\end{align*}
Let $\delta = (T/700 t) - 1$, and so $(1+\delta) = (T/700 t)$. Let $M= 100t$ be an upper bound of $\E[Y]$.  We can apply a variant of Chernoff bound (where one could replace $\E[Y]$ by $M$ if $M\geq \E[Y]$, see Corollary 4 of \cite{PettieS-trianglefree}) to obtain:
\begin{align*}
\Pr(Y > T/7) &= \Pr(Y > (T/700t) \cdot 100t)  \\
& = \Pr(Y > (1+\delta) \cdot M) \\
& \leq \Pr(Y > \E[Y] + \delta M) \\
&  \leq \left(\frac{e}{1+\delta} \right)^{(1+\delta) \cdot M} && \text{Corollary 4 of \cite{PettieS-trianglefree}}\\ 
&= \left(\frac{e}{1+\delta} \right)^{T/7} \\
&\leq \left(\frac{1}{e \Delta^{7}}\right)^{T/7}  && (1+\delta) = \frac{T}{700t} > e^2\cdot \Delta^{7}\\
&\leq \frac{1}{\poly(n)} \cdot \frac{1}{\Delta^{T}} && T = \Omega(\log n)~.	
\end{align*}

%\begin{align*}
%\Pr(Y > \E[Y] + \delta M)&\leq \left(\frac{e^{\delta}}{(1+\delta)^{(1+\delta)}} \right)^{M} && \text{Corollary 4 of \cite{PettieS-trianglefree}}\\
%&  \leq \left(\frac{e}{1+\delta} \right)^{(1+\delta) \cdot M} \\
%&= \left(\frac{e}{1+\delta} \right)^{T/7} \\
%&\leq \left(\frac{1}{e \Delta^{7}}\right)^{T/7}  && (1+\delta) = \frac{T}{700t} > e^2\cdot \Delta^{7}\\
%&\leq \frac{1}{\poly(n)} \cdot \frac{1}{\Delta^{T}} && T = \Omega(\log n)~.	
%\end{align*}
By Lemma \ref{lem:coupling}, $\Pr(X > T/7) \leq \Pr(Y > T/7 ) \leq 1/\poly(n) \cdot 1/\Delta^T$.

Suppose that the number of outwardly occupied edges in $E'(P)$ is at most $T/7$. Consider an occupied edge that is not outwardly occupied. Its adjacent blocking edge must be adjacent to an edge $e \in E'(P)$ that is outwardly occupied. Moreover, each $e \in E'(P)$ can only be adjacent to at most two blocking edges. Therefore, the number of occupied edges is at most $T/7 + 2 (T/7) +6000t = 3T/7 + 6000t$. The number of empty edges in $E'(P)$ must be at least $\lceil T / 2 \rceil - 3T/7 - 6000 t > T/15$. Thus, the probability that there are fewer than $T/15$ empty edges on $P$ is at most $1/(\poly(n) \cdot \Delta^{T})$.

There are at most $n \cdot \Delta^{T-1}$ paths of length $T$. By taking a union bound  all length-$T$ paths, the probability that there is a length-$T$ path with fewer than $T/15$ empty edges is at most $(n \cdot \Delta^{T-1} )/(\poly(n) \cdot \Delta^{T}) = 1/\poly(n)$. This also implies that the algorithm did not fail w.h.p. 
%Combined with the assumption that $H(q-1)$ holds w.h.p., we conclude that $H(q)$ holds with high probability.
\end{proof}

%\begin{lemma} \label{lem:coupling} For every $e \in E'(P)$, let $X_e = 1$ if $e$ is outwardly occupied by a blocking edge in an alternating path in $\mathcal{P}_2$ and $X_e = 0$ otherwise. We also let $\{ Y_{e} \}_{e \in E'(P)}$ be independent Bernoulli random variables with probability $100t/T$.  Let $X = \sum_{e \in E'(P)} X_e$ and $Y = \sum_{e \in E'(P)} Y_e$. For any $M \geq 0$, we have that $\Pr(X > M)  \leq \Pr(Y > M )$. \end{lemma}

\begin{proof}[Proof of Lemma \ref{lem:coupling}]

First, we create i.i.d.~random variables $\theta_{e,i}$ for each edge $e \in E'(P)$ and $1 \leq i \leq 2t$, where $\Pr(\theta_{e,i} = 1) = 75/T$. For each $e \in E'(P)$, we let $Y_e = \bigvee_{1 \leq i \leq 2t} (\theta_{e,i} = 1)$. Clearly, $\Pr(Y_e = 1) \leq 2t \cdot (75/T) = 150t/ T$ and $\{Y_e\}_{e \in E'(P)}$ are independent. 

The next step is to couple $\{X_e\}_{e \in E'(P)}$ with $\{\theta_{e,i}\}_{e,i}$. Recall that in our algorithm, each alternating path $\tilde{P} \in \mathcal{P}_2$ (that is processed) selects the blocking edge uniformly at random from $E_0(\tilde{P})$ (unoccupied edges at the execution when $\tilde{P}$ occurs). Now we consider a modified algorithm that has the same outcome with the original algorithm. It allows us to do the coupling more easily. Given $e \in E'(P)$, let $out_{\tilde{P}}(e)$ denote the set of edges that are in $E_0(\tilde{P})$ and are outgoing edges from any endpoint of $e$. Note that if $e,e' \in E'(P)$ and $e \neq e'$, then $out_{\tilde{P}}(e)$ and $out_{\tilde{P}}(e')$ are disjoint.

Let $E'_{\tilde{P}}(P) \subseteq E'(P)$ denote the set of edges with non-empty $out_{\tilde{P}}(e)$. Since $\tilde{P} \in \mathcal{P}_2$, $|E'_{\tilde{P}}(P)| \leq T/1000$. Also, since $|out_{\tilde{P}}(e)| \leq 4$, we have $\sum_{e \in E'_{\tilde{P}}(P) } |out_{\tilde{P}}(e)| \leq 4T/1000$.

The modified algorithm maintains a counter, $occurred$, which is initially set to $0$. Then it goes through each edge $e \in E'_{\tilde{P}}(P)$ and does the following: Throw a coin independently with success probability $|out_{\tilde{P}}(e)|/(|E_0(\tilde{P})| - occurred)$. If it succeeds, it selects a blocking edge randomly from $out_{\tilde{P}}(e)$ and then it is done with processing $\tilde{P}$. Otherwise, it sets $occurred \leftarrow occurred + |out_{\tilde{P}(e)}|$ and continues to process the next edge in $E'_{\tilde{P}}(P)$. If the algorithm did not successfully select a blocking edge after it processed every edge in $E'_{\tilde{P}}(P) \subseteq E'(P)$, it selects an edge from $E_0(\tilde{P}) \setminus (\bigcup_{e \in E'_{\tilde{P}}(P)} out_{\tilde{P}}(e))$ uniformly at random as the blocking edge.

It is easy to see that in the modified algorithm, the probability that any edge in $E_0(\tilde{P})$ is selected as the blocking edge is exactly $1 /|E_0(\tilde{P})|$. Therefore, the modified algorithm is equivalent to the original algorithm where $\tilde{P}$ selects a random edge in $E_0(\tilde{P})$ as the blocking edge.

\sloppy
Now we couple the randomness of the modified algorithm with $\{\theta_{e,i}\}_{i,e}$. When the modified algorithm processes an edge $e \in E'_{\tilde{P}}(P)$, instead of throwing a new coin with probability $|out_{\tilde{P}}(e)|/(|E_0(\tilde{P})| - occurred)$, we will couple it with the randomness from $\theta_{e,i}$. For each $e \in E'(P)$, we maintain a counter $q(e)$ that points to the smallest index such that $\theta_{e, q(e)}$ is not yet exposed. Initially, $q(e)$ are 1 for every $e \in E'(P)$. Now when the algorithm processes $\tilde{P} \in \mathcal{P}_2$ and an edge $e \in E'_{\tilde{P}}(P)$, we will expose $\theta_{e, q(e)}$. Let $\mathcal{E}_1$ denote the event that $\theta_{e, q(e)} = 1$. Also, we set $q(e) \leftarrow q(e) + 1$.  Let $\mathcal{E}_2$ denote the event that a coin with success probability $ \frac{1}{ \Pr(\mathcal{E}_1)} \cdot \frac{|out_{\tilde{P}}(e)|}{(|E_0(\tilde{P})| - occurred)}$ succeeds.  Note that $\Pr(\mathcal{E}_1) \cdot \Pr(\mathcal{E}_2) = |out_{\tilde{P}}(e)|/(|E_0(\tilde{P})| - occurred)$. The original coin with probability  $|out_{\tilde{P}}(e)|/(|E_0(\tilde{P})| - occurred)$ will be simulated by the trial that both the two events $\mathcal{E}_1$ and $\mathcal{E}_2$ happen. 

Note that given $e \in E'(P)$, at most $2t$ alternating paths in $\mathcal{P}_2$ may be adjacent to $e$ in total. This implies $\theta_{e, q(e)}$ is always well-defined when the algorithm refers to it. It remains to show that $\Pr(\mathcal{E}_2)$ is a valid probability, namely, $|out_{\tilde{P}}(e)|/(|E_0(\tilde{P})| - occurred) \leq \Pr(\mathcal{E}_1)$. 
\begin{align*}
\Pr(\mathcal{E}_1) &= 75 / T \\
%&\geq 1 - \left(1 - \frac{13}{15} \cdot \frac{75 |out_{\tilde{P}}(e)|}{ T} \right)&& (1-x)^n \leq 1-\frac{13}{15}nx\\
%& && \mbox{for $nx$ sufficiently small}  \\
%&= \frac{65 |out_{\tilde{P}}(e)|}{T} \\
&\geq \left( \frac{|out_{\tilde{P}}(e)|}{T/15 - 4T/1000} \right) && |out_{\tilde{P}}(e)| \leq 4\\
&\geq  \left( \frac{|out_{\tilde{P}}(e)|}{T/15 - occurred} \right) && occurred \leq 4T/1000 \\
&\geq \left( \frac{|out_{\tilde{P}}(e)|}{|E_0(\tilde{P})| - occurred} \right) && |E_0(\tilde{P})| \geq T/15~.
\end{align*}

The second to the last line holds since $occurred \leq \sum_{e \in E'_{\tilde{P}}(P) } |out_{\tilde{P}}(e)| \leq 4T/1000$. Also, the last line holds since if $|E_0(\tilde{P})| < T/15$, the algorithm freezes and $\tilde{P}$ will not be processed.  

Finally, we point out that the coupling argument implies whenever $X_e = 1$, it must be the case that $\theta_{e,i} = 1$ for some $1 \leq i \leq 2t$. This in turns implies that $Y_e = 1$ and thus $X_e \leq Y_e$. Therefore, for any $M>0$, $\Pr(X > M) \leq \Pr(Y > M)$.
\end{proof}

}
By Lemma \ref{lem:outerloop}, Line \ref{ln:low_arbor_color} is executed $O(\Delta^4 \log n)$ times. Each execution uses $O(T \cdot \log n)$ time, where $T = O(\Delta^{11} \log n)$. By Lemma \ref{lem:outerloop} again, Line \ref{ln:last} is also executed $O(\Delta^4 \log n)$ times. Each execution takes $O(T)$ time. Therefore, the number of rounds used by our algorithm to obtain a $(\Delta+2)$-edge-coloring is $O(\Delta^{15}\log^3 n)$. The following lemma (whose proof we postpone to the end of the section) completes the proof of Theorem \ref{thm:randomized}, which shows the algorithm can be converted to a $(1+\epsilon)\Delta$-edge-coloring algorithm that runs in $O((1/\epsilon)^{15} \log^{3} n)$ rounds for any $\epsilon \geq 2 / \Delta$.

\begin{lemma}\label{lem:convert}Suppose that there is a deterministic (randomized) $(\Delta + z)$-edge-coloring algorithm that runs in $O(\Delta^x \log^{y} n)$ rounds. Then there exists a deterministic (randomized) $(1+\epsilon)\Delta$-edge-coloring algorithm that runs in $O((1/\epsilon)^x \cdot (z+1)^{x} \cdot \log^{y} n)+\tilde{O}(\log^2\Delta \cdot \epsilon^{-1} \cdot \log \epsilon^{-1} \cdot \log n)$ rounds for any $\epsilon \geq z/\Delta$, where $\tilde{O}(\cdot)$ hides a $\poly(\log \log n)$ factor.\end{lemma}

%We postpone the proof of Lemma \ref{lem:convert} after we show how to achieve the load of $O(\log n / \log \log n)$ deterministically.

\paragraph{Achieving $\ell(G) = O(\log n / (\log \lambda + \log \log n))$ deterministically} We will choose $T = 2t\cdot \lambda$. Let $(1+\delta) = K \cdot \lambda \cdot \log n /(\log \lambda + \log \log n)$ for some constant $K > 0$.  Given an alternating path $P$ with length at least $T$, let $\ell_q (v)$ denote the load of $v$ after $q$'th execution of Line \ref{ln:last} of Algorithm \ref{alg:main}.

During the $q$'th execution of Line \ref{ln:last} of Algorithm \ref{alg:main}, we select the blocking edge as follows. Given an edge $e = uv$, define the cost $c(e)$ to be $(1+\delta)^{\ell_{q-1}(v)} + (1+\delta)^{\ell_{q-1}(u)} $. Given an alternating path $P$, we select the edge $e \in P(T)$ with the minimum cost $c(e)$ to be the blocking edge. We will show that this greedy algorithm achieves a maximum load of $O(\log n /(\log \lambda + \log \log n ) )$.

\begin{lemma} If $T = 2t \lambda$, the greedy algorithm achieves a load of $O(\log n / (\log \lambda+ \log \log n))$. \end{lemma}
\begin{proof}
We use the method of conditional expectation to show that the greedy algorithm is a derandomization of the randomized algorithm. Define the pessimistic estimator  $$\Phi(q) = \sum_{v \in V} \frac{(1+\delta)^{\ell_q(v)} \cdot (1+\delta \cdot (2 / T) )^{t-q}}{(1+\delta)^{(1+\delta)/\lambda}}~.$$

We show that $\Phi(q+1) \leq \Phi(q)$  by the choice of the greedy algorithm. Consider the $(q+1)$'th iteration of Line \ref{ln:last} of Algorithm \ref{alg:main}.  

Let $B$ denote a set of blocking edges chosen during $(q+1)$'th execution. Given an edge $e$, if $v$ is an endpoint of $e$, we write $v \in e$. Moreover, we write $v \in B$, if there exists $e \in B$ such that $v \in e$. We use $[ v \in B]$ denote the indicator variable for the event $v \in B$. Let $Q(B)$ denote the following quantity:

$$Q(B) = \sum_{v \in V} (1+\delta)^{\ell_q(v) + [v \in B]}~.$$

Suppose that we choose each blocking edge randomly in $(q+1)$'th execution of Line \ref{ln:last}, we have \begin{align*}  & \E[Q(B)] \\ 
&= \sum_{v \in V} \left( \Pr(v \in B)\cdot (1+\delta)^{\ell_q(v)+1} + \Pr(v \notin B) \cdot (1+\delta)^{\ell_q(v)}  \right) \\
&= \sum_{v \in V}  (1+\delta)^{\ell_q(v)} \cdot \left(\Pr(v \in B)\cdot(1 + \delta) + (1 - \Pr(v \in B)) \right)\\
&=  \sum_{v \in V}  (1+\delta)^{\ell_q(v)} \cdot \left(1 + \delta \cdot \Pr(v \in B) \right) \\
&\leq  \sum_{v \in V} (1+\delta)^{\ell_q(v)} \cdot \left(1 + \delta \cdot \frac{2}{T}  \right) ~. \end{align*}

Observe that given $e$ and $e'$, if $c(e) < c(e')$, then
$\sum_{v \in V} (1+\delta)^{\ell_q(v) + [v \in e]} < \sum_{v \in V} (1+\delta)^{\ell_q(v) + [v \in e']}$. Moreover, the alternating paths generated in the $(q+1)$'s execution of Line \ref{ln:last} do not intersect. Therefore, our greedy strategy is always choosing a set of blocking edges to minimize $Q(B)$. Thus,

\begin{align*}
\Phi(q + 1) &= \frac{(1+\delta \cdot (2/T))^{t-q-1}}{(1+\delta)^{(1+\delta)/\lambda}} \cdot \min_{B} Q(B) \\
&\leq \frac{(1+ \delta \cdot (2/T))^{t-q-1}}{(1+\delta)^{(1+\delta)/\lambda}} \cdot \E_{B}[Q(B)] \\
&\leq \frac{(1+\delta  \cdot (2/T))^{t-q-1}}{(1+\delta)^{(1+\delta)/\lambda}} \cdot \sum_{v \in V} (1+\delta)^{\ell_q(v)}\left(1 + \delta \cdot (2/T) \right) \\
& = \Phi(q)~.
\end{align*} 

Let $X = \log \lambda + \log \log n$. We have
\begin{align*}
 &\sum_{v \in V} (1+\delta)^{\ell_t(v) - (1+\delta)/\lambda} = \Phi(t) \leq \Phi(t-1) \leq \ldots \leq \Phi(0) \\
 &= \sum_{v \in V} \frac{(1+\delta \cdot (2/T))^{t}}{(1+\delta)^{(1+\delta)/\lambda}} \\
&\leq \sum_{v \in V} \frac{e^{\delta \cdot (2t / T)}}{(1+\delta)^{(1+\delta)/\lambda}}&& \text{(since  $1 +x \leq e^{x}$)} \\
&\leq  \sum_{v \in V} \frac{e^{\delta / \lambda}}{(1+\delta)^{(1+\delta)/\lambda}} \\
&\leq \sum_{v \in V} \left[ \frac{e}{K \lambda \log n / (\log \lambda + \log \log n)} \right]^{K\log n / (\log \lambda + \log \log n)} \\
&= \sum_{v \in V}\exp(-(X + \log K - \log(X) - 1) \cdot K \log n / X ) \\
&\leq \sum_{v \in V}\exp(-(K/2) \log n) && \text{(for sufficiently large $X$)}\\
&\leq \sum_{v \in V} \frac{1}{\poly(n)} < 1~.
\end{align*}
This implies $\ell_t(v) \leq (1+\delta)/\lambda = O(\log n / (\log \lambda + \log \log n))$ for all $v$. Thus, we have achieved $\ell(G) = O(\log n /(\log \lambda + \log \log n))$.
\end{proof}
The running time is dominated by Line \ref{ln:low_arbor_color}, which is executed $O(\Delta^4 \log n)$ times. Each execution uses $O(T \cdot \log n)$ time. The total running time for obtaining a $\Delta + O(\log n / (\log \lambda + \log \log n))$ coloring is $O(\Delta^4 T \cdot \log^2 n) = O(\lambda \cdot \Delta^8 \log^3 n)$. We show how to reduce the running time to $O(\lambda \cdot \Delta^6 \log^3 n)$ in the next section. By Lemma \ref{lem:convert}, it can be converted to a deterministic $(1+\epsilon)\Delta$-edge-coloring algorithm that runs in $O(\lambda \cdot  (1/\epsilon)^6 \log^9 n)$ algorithm for $\epsilon = \Omega(\log n / (\Delta (\log \lambda  +\log \log n)))$.

\begin{proof}[Proof of Lemma \ref{lem:convert}]
W.L.O.G.~we may assume $\epsilon \leq 1/8$. If $\Delta \leq (16(z+8)) / \epsilon$, then by applying the $(\Delta+z)$-edge-coloring algorithm, we obtain an $(1+\epsilon)\Delta$-edge-coloring in $O(\Delta^{x} \log^{y} n) = O((1/\epsilon)^{x} \cdot (z+1)^{x} \log^{y} n)$ rounds, since $z \leq \epsilon \Delta$ and $\epsilon \leq (16(z+8))/\Delta$. 

Otherwise, let $\epsilon' = \epsilon / (32 \log_{4/3} \Delta)$.  By using the splitting algorithm of \cite{GHKMSU17}, we obtain two subgraphs of maximum degree $(1/2 + \epsilon')\Delta + 4$ in $\tilde{O}(\epsilon'^{-1} \cdot \log \epsilon'^{-1} \cdot \log n)$ rounds. Let $\Delta_0 = \Delta$ and $\Delta_i = (1/2 + \epsilon')\Delta_{i-1} + 4$. If we apply the algorithm recursively for $h$ levels, we obtain $2^h$ subgraphs whose maximum degree are bounded by $\Delta_h$. We choose $h$ to be the smallest number that $(1/2 + \epsilon')^{h}\cdot \Delta \in [(8(z+8))/\epsilon, (16(z+8))/\epsilon]$. Note that $h \leq \log_{4/3} \Delta$, since $\epsilon' \leq \epsilon \leq 1/4$. Note that 
\begin{align*}
\Delta_h &= (1/2 + \epsilon')^{h}\cdot \Delta + \sum_{i=0}^{h-1} (1/2 + \epsilon')^{i} \cdot 4 \\
&\leq (1/2 + \epsilon')^{h}\cdot \Delta + \sum_{i=0}^{h-1} (3/4)^{i} \cdot 4   && \text{(since  $\epsilon' < 1/4$)}\\
&\leq (1/2 + \epsilon')^{h}\cdot \Delta + 16 \\
&\leq (1+\epsilon/4) (1/2 + \epsilon')^{h}\cdot \Delta \\
& \hspace{1cm} \text{(since $(1/2+\epsilon')^{h} \cdot \Delta \geq (8(z+8))/\epsilon \geq 64/\epsilon$)}~.
\end{align*}

%Also, since $\Delta_{i-1}\geq \Delta_h  \geq (8(z+1))/\epsilon' \geq 4/ \epsilon'$, $\Delta_i = (1/2 + \epsilon')\Delta_{i-1} + 4 \leq (1/2 + 2\epsilon')\Delta_{i-1}$.

We run the $(\Delta_h + z)$-edge-coloring algorithm on the $2^h$ subgraphs using disjoint palettes to obtain a $(\Delta_h+z)$-edge-coloring for each. The total number of colors that are used is
\begin{align*}
2^h \cdot (\Delta_h+z) &\leq 2^h \cdot (1+\epsilon / 4)\cdot \Delta_h && \text{(since $\Delta_h \geq (1/2 + \epsilon')^{h} \cdot \Delta \geq 8z / \epsilon$)} \\
&= 2^h \cdot (1+\epsilon / 4)^2 \cdot (1/2 + 2\epsilon')^{h} \cdot \Delta  && \text{(since $\Delta_h \leq (1+\epsilon/4)\cdot(1/2 + 2\epsilon')^{h} \cdot \Delta$)}\\
&= (1+\epsilon / 4)^2 \cdot (1 + 4\epsilon')^{h} \cdot \Delta  \\
&\leq (1+\epsilon / 4)^2 \cdot (1 + 8h \epsilon') \cdot \Delta && \text{(since \mbox{$(1+a)^b \leq 1+2ab$ for $0 \leq ab \leq 1$})} \\
&\leq (1+\epsilon / 4)^2 \cdot (1 + \epsilon / 4) \cdot \Delta && \text{(since $\epsilon' = \epsilon / (32 \log_{4/3} \Delta) \leq \epsilon / (32 h)$)}\\
&\leq (1+\epsilon)\Delta && \text{(when $\epsilon < 1/8$)} ~.
\end{align*}

Since $\Delta_h = O(z/\epsilon)$, the running time is $\tilde{O}(\log\Delta \cdot \epsilon'^{-1} \cdot \log \epsilon'^{-1} \cdot \log n) + O((1/\epsilon)^{x} (z+1)^{x} \log^{y} n)  $.
\end{proof}

%!TEX root = main.tex

\section{Faster Algorithms} 
In this section, we provide more efficient algorithms that use fewer number of rounds. Our technique is based on another type of fans that was introduced by Gabow et al. \cite{GNKLT85}. This  allows us to reduce the number of rounds by a factor $\Delta^2$ for general graphs and $\Delta^4$ for bipartite graphs.

%$\alpha \notin {B}$ and 

A \emph{reverse $\alpha {B}$-fan} with center $v$ and leaves $x_1,x_2,\ldots,x_k$ (where ${B}=\{ \beta_1,\beta_2,\ldots,\beta_{k+1} \}$ and  $k \geq 2$)  satisfies:  All edges $vx_i$ are uncolored.  The color $\alpha$ is missing at all leaves, i.e., $\alpha \in M(x_i)$ for all $1 \leq i \leq k$.  The colors in $B$ are missing at $v$, i.e., $\beta \in M(v)$ for all $\beta \in B$.  Finally, there is an edge with color $\alpha$ that is incident to the center $v$, i.e., $\alpha \notin M(v)$. The above definition is well-defined since there must be at least $k+1$ missing colors at $v$. We use \emph{normal fans} and \emph{reverse fans} to distinguish between the two types. We use fans to refer to fans of either types. We use \emph{incomplete vertices} to refer to vertices that are incident to at least one uncolored edge.

\begin{algorithm}
\caption{Faster Distributed Fan Repair}\label{alg:fast}
\begin{algorithmic}[1]\small
 \While{there exists uncolored edges}  \label{ln:while-loop-2}
	
%	\For{$\eta = $ color $1, 2, \ldots, \Delta + 1$} \label{ln:iteration-eta-2}
%		\State Let $E_{\eta} = \{ uv \in E \mid \mbox{$uv$ is uncolored and $\eta \in M(u) \cap M(v)$} \}$.
%		\State Let $M_{\eta}$ be a maximal matching  of $E_{\eta}$.
%		\State Color $M_{\eta}$ with color $\eta$. \label{ln:color_maximal}
%	\EndFor
	
	\State Color the graph using $\Delta+1$ colors.

	\For{$\gamma =$ color $1,2,\ldots,\Delta+1$}  \label{ln:iteration-gamma-2}
		\State Let $V_\gamma$ be the vertices with color $\gamma$.

		\For{$\alpha = $ color $1, 2, \ldots, \Delta + 1$} \label{ln:iteration-alpha-2}
			\State Let $E_{\gamma,\alpha} = \{ uv \in E \mid \mbox{$u \in V_\gamma$ and $uv$ is uncolored and $\alpha \in M(u) \cap M(v)$} \}$.
			\State Let $M_{\gamma,\alpha}$ be a maximal matching  in $E_{\gamma,\alpha}$.
			\State Color $M_{\gamma,\alpha}$ with $\alpha$. \label{ln:color_maximal}
			
			\State Let $V_{\gamma,\alpha}$ be the set of incomplete vertices in $V_\gamma$ that now miss the color $\alpha$.  
			\State Color $G^{4}[V_{\gamma,\alpha}]$ using at most $4\Delta^4$ colors.	 \label{ln:merge-fans}	

			\For{$i = 1, 2, \ldots, 4\Delta^4$}
				\State Grow a normal $\alpha$-fan $F_v$ from each vertex $v$ with color $i$ in $G^{ 4}[V_{\gamma,\alpha}]$. \label{ln:grow-fans-2}
				\State $\merge $  intersecting fans.  \label{ln:end-merge-fans}  
			\EndFor 

			 \For{$\beta = $ color $1, 2, \ldots, \Delta + 1$} \label{ln:iteration-beta-2}  \label{ln:maximal-disjoint}  
				\State Build a conflict graph $G_{\mathcal{F}_{\gamma,\alpha, \beta}}$ for normal and sub-reverse $\alpha\beta$-fans in $\mathcal{F}_{\gamma,\alpha,\beta}$.
				\State Color $G_{\mathcal{F}_{\gamma,\alpha, \beta}}$ using $O(1)$ colors.
				\For{$\eta =$ color $1,2,\ldots O(1)$}
					\State Repair undestroyed normal and sub-reverse fans in $\mathcal{F}_{\gamma,\alpha, \beta}$ with color $\eta$ in $G_{\mathcal{F}_{\gamma,\alpha, \beta}}$ . \label{ln:parallel-agumentation-2a}
				\EndFor
				\State Repair semi-destroyed sub-reverse  fans in $G_{\mathcal{F}_{\gamma,\alpha, \beta}}$. \label{ln:parallel-agumentation-2b}
			\EndFor
		\EndFor		
	\EndFor
\EndWhile
\end{algorithmic}
\end{algorithm}

\paragraph{Algorithm outline} We outline our algorithm (Algorithm \ref{alg:fast}) and analysis as follows. Similar to the previous section, we will show that each iteration of the {\bf while loop} in Line \ref{ln:while-loop-2} will color at least a constant fraction of the uncolored edges and therefore the algorithm will terminate after  $O(\log n)$ iterations. Roughly speaking, we iterate through the incomplete vertices that miss the colors $\alpha =1,2,\ldots,\Delta+1$ (i.e., the {\bf for loop} in Line \ref{ln:iteration-alpha-2}). For each color $\alpha$, we try to make a maximal set of vertex-disjoint $\alpha$-fans (i.e., every vertex in $V_{\gamma,\alpha}$ belongs to a fan) that consists of normal fans and reverse fans. This is done by merging intersecting  fans in a clever way.
Then, the algorithm tries to repair these fans simultaneously. The  idea is that we can repair the fans simultaneously in fewer rounds since they are vertex-disjoint. There are two new key challenges in the \local model (beside long alternating paths):
\begin{itemize}
\item We need to implement the merge step in the  distributed setting efficiently. Gabow et al. \cite{GNKLT85} only provided a sequential algorithm for this.
\item With two types of fans, there are more cases to analyze when repairing a fan destroys the structure of another fan. This can become tricky since a destroyed reverse fan may lead to many edges remain uncolored.
\end{itemize}

\subsection{Making a Maximal Set of $\alpha$-Fans}

We consider a fixed iteration of the outermost {\bf while loop}. We first color the vertices using $\Delta+1$ colors which could be done deterministically in  $O(\Delta + \log^{*} n)$ rounds \cite{BEK09}. We then process vertices in the same color class $\gamma$, denoted by $V_\gamma$, together.  The reason for this step will become clear later.

%$M_{\gamma,\alpha}$ would not be a maximal matching. 

%For each color $\eta=1,2,\ldots,\Delta+1$, we find a maximal matching $M_{\gamma,\eta}$ among the edges $uv$ (where $u,v \in V_\gamma$) in which both $u$ and $v$ miss the color $\eta$ in $O(\Delta + \log n)$ rounds \cite{BEK09}. We then color $M_{\gamma,\eta}$ with the color $\eta$. 

Fix a nested iteration $\gamma$ and $\alpha$  of the {\bf for loop} in Line \ref{ln:iteration-gamma-2} and the {\bf for loop} in  Line \ref{ln:iteration-alpha-2}. We first consider the set $E_{\gamma,\alpha}$ of all uncolored edges $uv$  that are incident to one vertex in $V_{\gamma}$ such that both $u$ and $v$ miss the color $\alpha$. More formally,
$$E_{\gamma,\alpha} = \{ uv \in E \mid \mbox{$u \in V_\gamma$ and $uv$ is uncolored and $\alpha \in M(u) \cap M(v)$} \}~.$$

We find a maximal matching $M_{\gamma,\alpha}$ in $E_{\gamma,\alpha}$ and color the matching with the color $\alpha$ in $O(\Delta + \log n)$ rounds \cite{BEK09}.

After coloring $M_{\gamma,\alpha}$, recall from Algorithm \ref{alg:fast} that $V_{\gamma,\alpha}$ denotes the set of incomplete vertices in $V_\gamma$  that miss the color $\alpha$.  Note that $V_{\gamma,\alpha}$ is an independent set since the $\gamma$-color vertices must not be neighbors.

We now describe a sequential  procedure that finds a maximal set of $\alpha$-fans. Specifially, each vertex in $V_{\gamma,\alpha}$ is either a center of a fan or a leaf of a reverse fan. Furthermore, all the fans are vertex-disjoint and the number of uncolored edges remain the same.  

We want to maintain the following invariant: a) all active fans are vertex-disjoint and b) all vertices $v \in V_{\gamma,\alpha}$ that have been considered must either be a center of a normal fan or a leaf of a reverse fan.  Now, suppose we consider the next vertex $v$ in $V_{\gamma,\alpha}$ and grow a normal $\alpha$-fan from $v$.  We have argued that $V_{\gamma,\alpha}$ is an independent set. Assuming the invariant holds up to this point, we have:

\begin{claim} \label{claim:intersecting-fans}
Consider a normal $\alpha$-fan $F_v$ grown  from $v \in V_{\gamma,\alpha}$ that intersects with a currently active $\alpha$-fan $F_u$. If $F_u$ is a normal fan, they must intersect at leaf node(s). On the other hand, if $F_u$ is a reverse  fan, then they only intersect at $u$. 
\end{claim}

Suppose $F_v$ intersects with at least one other $\alpha$-fan. Let the leaves of $F_v$ be $x_1,\ldots,x_k$ and let $w=x_j$ be the first  leaf that is in an intersection, i.e., $j$ is the smallest possible. Suppose $F_v$ intersects with $F_u$ at $w$. 

If $F_u$ is a normal fan, then we shift $F_v$ and $F_u$ from $w$ and uncolor the edges $vw$ and $uw$. We  then deactivate $F_u$ and $F_v$ and grow a reverse fan centering at $w$ with the leaves $u$ and $v$. We note that $w$ cannot miss $\alpha$ otherwise the matching $M_{\gamma,\alpha}$  is not maximal. 

If $F_u$ is a reverse fan, then  note that $u=w$ based on Claim \ref{claim:intersecting-fans}. We shift $F_v$ from $u$ and uncolor the edge $vu$. We then deactivate $F_v$ and add $vu$ to  $F_u$. 

The fact that $w$ is the first leaf that is in an intersection is important since $F_v$ may intersect with many current active fans and shifting at another leaf vertices may destroy other fans that are not $F_v$ and $F_u$. It is easy to see that the invariant that the number of uncolored edges remain the same and the fans are vertex-disjoint  holds after we process each vertex $v$ in $V_{\gamma,\alpha}$. Hence, we obtain a maximal set of vertex-disjoint $\alpha$-fans.

\paragraph{Merging fans the distributed setting}  We still need to implement this idea in the  \local model. We now describe the merge step that corresponds to Line \ref{ln:merge-fans} to Line \ref{ln:end-merge-fans} in our algorithm. 

First, we obtain a $4$-hop coloring of $V_{\gamma,\alpha}$ using at most $4\Delta^4$ colors which can be done in $O(\Delta^4 + \log^{*} n)$ rounds \cite{BEK09}. In particular, this is a vertex coloring of the graph $G^{4}[V_{\gamma,\alpha}]$ in which there is an edge between $uv$ if the distance between $u$ and $v$ is less than five in $G$. We then go through vertices in each color class $i=1,\ldots,4\Delta^4$ and grow a normal $\alpha$-fan from vertices in that class. 

The key observation is as follows. Suppose $u$ and $v$ are from the same color class $i$. Since $\dist_G(u,v) > 4$,  the normal fans $F_u$ and $F_v$ that are grown in Line \ref{ln:grow-fans-2} do not intersect and furthermore, $F_u$ and $F_v$ cannot intersect the same currently active $\alpha$-fan $F_w$. Therefore, we can simultaneously merge the newly grown fans with the existing active fans as described in a constant number of rounds (Algorithm \ref{alg:merge}).

\begin{algorithm}
\caption{$\merge$}\label{alg:merge}
\begin{algorithmic}[1]\small
	\State Let $A$ be the set of fans that have just been grown from vertices in color class $i$ of the $4$-hop coloring.
	\State Let $B$ be the set of current active fans including the fans in $A$.
	\For{each $F_v$ in $A \setminus B$ that intersects with another fan in $B$}
		\State  Let $x$ be the first  leaf of $F_v$ that is in an intersection.
		\State Let $F_w$ be the fan that intersects with $F_v$ at $x$.
		\If{$F_w$ is a normal fan}
			\State Shift $F_w$ and $F_v$ from $x$.
			\State Uncolor $wx$ and $vx$.
			\State Grow a reverse $\alpha$-fan centering at $x$ with leaves $w$ and $v$.
			\State Deactivate $F_v$ and $F_w$.
		\Else{}
			\State Shift $F_v$ from $x$.
			\State Uncolor $vx$.
			\State Add $vx$ to $F_w$ (note that $w=x$).
			\State Deactivate $F_v$.
		\EndIf

	\EndFor

\end{algorithmic}
\end{algorithm}

Hence, we obtain a maximal set of vertex-disjoint $\alpha$-fans after we iterate through  all color classes $i=1,2,\ldots,4\Delta^4$.  We summarize this as the following claim.

\begin{claim}
After the merge step (just before Line \ref{ln:maximal-disjoint}), we have a maximal set of vertex-disjoint $\alpha$-fans with respect to $V_{\gamma,\alpha}$. 
 \end{claim}

\subsection{Repairing Fans}

Recall that we currently fix a nested iteration $\gamma$ and $\alpha$  of the for loops in Line \ref{ln:iteration-gamma-2} and Line \ref{ln:iteration-alpha-2}. We furthermore have a maximal set of vertex-disjoint $\alpha$-fans.

\paragraph{Sub-reverse  fans}   For the sake of a simpler analysis, we further decompose a reverse $\alpha B$-fan that $F_w$ with $k$ leaves into $\floor{k/2}$ {edge-disjoint} \emph{sub-reverse fans}.   The first sub-reverse fan consists of $\{w,x_1,x_2 \}$. The second sub-reverse  fan consists of $\{w,x_3,x_4 \}$, and so on. The last sub-reverse fan may consist of three leaves if $k$ is odd.  

Recall that $B=\{\beta_1,\beta_2,\ldots,\beta_{k+1}\}$ is a set of $k+1$  colors missing at $w$. We associate $\beta_1$ with the first sub-reverse  fan, $\beta_2$ with the second sub-reverse  fan, and so on. We use $F_{w,1},F_{w,2},\ldots,F_{w,\floor{k/2}}$ to denote these sub-reverse fans. In particular, we call $F_{w,i}$ a sub-reverse $\alpha \beta_i$-fan. 

We now focus on normal fans and sub-reverse fans. Let $\alpha \beta$-fans denote normal $\alpha \beta$-fans or sub-reverse $\alpha \beta$-fans. We consider a simple repair for a sub-reverse $\alpha \beta_i$-fan with center vertex $w$ and (at least) two leaf vertices $x_{2i},x_{2i+1}$.  Consider the maximal alternating $ \alpha \beta_i$-path $P_w$ from the center $w$.   We can first  $\augment(P_w)$.  If $P_w$ does not end at $x_{2i}$, then color $wx_{2i}$ with $\alpha$. Otherwise,  if $P_w$ ends at $x_{2i}$, then color $wx_{2i+1}$ with $\alpha$. This repair will color  an uncolored  edge. 
After the algorithm repairs $F_{w,i}$, then either $w x_{2i}$ or $w x_{2i+1}$ will be colored with the color $\alpha$. Therefore, we still have $\alpha \notin M(w)$. As a result, other sub-reverse fans centering at $w$ are not destroyed.

\begin{claim}\label{claim:repair-sub-reverse-fans}
Repairing a sub-reverse fan $F_{w,i}$ does not destroy another sub-reverse fan $F_{w,j}$.
\end{claim}

\paragraph{Repairing  fans in the distributed setting}  As outlined in the previous section, if the length of a maximal alternating path $P_v = (v_0,v_1,\ldots)$ is longer than $T$, we truncate $P_v$ and pick a blocking edge $(v_{i-1},v_i)$ accordingly (i.e., $i \leq T$ is chosen based on whether the goal is a deterministic algorithm with $\ell(G) = O(\log n / (\log \lambda + \log \log n))$ or a randomized algorithm with $\ell(G) = 1$). We then $\augment(P_v(i))$.

We make the following observation: Based on the definition, a vertex in a sub-reverse $\alpha \beta$-fan misses either $\beta$ if it is the center or $\alpha$ if it is a leaf. Therefore, a maximal  $\alpha \beta$-path cannot intersect a sub-reverse $\alpha \beta$-fan except at the last vertex of the path. Thus, a truncated $\alpha \beta$-path $P_v(T)$ where $|P_v| > T$ cannot intersect with a sub-reverse $\alpha \beta$-fan.

Next, let us fix an iteration $\beta$  in the {\bf for loop} of Line \ref{ln:iteration-beta-2}. We again observe that the alternating paths of  $\alpha \beta$-fans do not intersect. Augmenting an alternating path can only destroy the structure of at most one other  fan. This happens when an alternating path of length at most $T$  ends at another fan.

Let the set of the current normal and sub-reverse $\alpha \beta$-fans   be $\mathcal{F}_{\gamma,\alpha,\beta}$.  We build a conflict graph on $G_{\mathcal{F}_{\gamma,\alpha,\beta}} = (V_{\mathcal{F}_{\gamma,\alpha,\beta}}, E_{\mathcal{F}_{\gamma,\alpha,\beta}})$, where $V_{\mathcal{F}_{\gamma,\alpha,\beta}} = \{F_v \mid F_v \in \mathcal{F}_{\gamma,\alpha,\beta} \}$ and $(F_u, F_w) \in E_{\mathcal{F}_{\gamma,\alpha,\beta}}$ if $|P_u| \leq T$ and $P_u$ ends at a vertex in $F_w$. If $|P_u| > T$ and the truncated path $P_u(T)$ ends at any node of $F_w$, then $F_w$ must be a normal fan and furthermore the modification on $P_u(T)$ will not affect the structure of $F_w$ as pointed out in the previous section.

Since each vertex in $G_{\mathcal{F}_{\gamma,\alpha,\beta}}$ contributes at most one edge, the arboricity of $G_{\mathcal{F}_{\gamma,\alpha,\beta}}$ is constant. Therefore, as argued in the previous section, we can build and color the conflict graph using $O(1)$ colors in $O(T \log n)$ rounds \cite{barenboim2010sublogarithmic}. Now we iterate through each color $j$ of the coloring and process the fans colored in $j$ together.  Since the alternating paths of the fans colored in $j$ neither intersect nor end at another fan, we can repair them simultaneously. 

There are three  cases to take into consideration when an alternating path of length at most $T$ ends at  another fan $F_w$:
\begin{itemize}
\item Case 1: $F_w$ is a normal $\alpha \beta$-fan. 
\item Case 2: $F_w$ is a sub-reverse $\alpha \beta$-fan and the alternating path ends at a leaf of $F_w$.
\item Case 3: $F_w$ is a sub-reverse $\alpha \beta$-fan and the alternating path ends at $w$.
\end{itemize}

We say that $F_w$ is destroyed if Case (1) or Case (2) happens. This is because augmenting that path will destroy the structure of $F_w$.

We however cannot  ignore Case (3) since it may destroy $\Omega(\Delta)$ sub-reverse fans centering at $w$. If case (3) occurs, we say that $F_w$ is \emph{semi-destroyed}. Fortunately, it is relatively easy to repair all semi-destroyed fans simultaneously.

As argued above, if Case (3) happens, the final edge of the alternating path must be a color-$\alpha$ edge. After the alternating path is augmented,  the final edge will become a color-$\beta$ edge and therefore $\alpha$ becomes missing at $w$. Hence, we can color an edge in $F_w$ with $\alpha$. We call this step repairing semi-destroyed sub-reverse fans. Note that this step can be done simultaneously among $\alpha \beta$-fans.  Based on Claim \ref{claim:repair-sub-reverse-fans}, no other sub-reverse fan centering at $w$ is destroyed during this step. We also note that each $\shift$ operation to repair a normal fan will not destroy another $\alpha$-fan since all the fans are vertex-disjoint.

\paragraph{Analysis} Now, we will show that each iteration in the while loop  colors at least a constant fraction of the uncolored edges. Let $\phi_0$ be the partial coloring at the beginning of a fixed interation of the while loop. Let $\phi_{\gamma}$ be the partial coloring just before the iteration $\gamma$ of the { for loop} in Line \ref{ln:iteration-gamma-2}. Similarly, let $\phi_{\gamma,\alpha}$ be the partial coloring just before the iteration $\gamma$ and $\alpha$ of the nested for loops in Line \ref{ln:iteration-gamma-2} and Line \ref{ln:iteration-alpha-2}. We again use $|\phi_0^{-1}(\bot)|$ to denote the number of uncolored edges at the beginning of current iteration of the while loop.

We use $\mathcal{F}_{\gamma,\alpha,\beta}$ to denote the normal and sub-reverse $\alpha \beta$-fans considered in the nested iteration $\gamma,\alpha,\beta$.  In addition, let $\mathcal{F}_{\gamma,\alpha} = \cup_{\beta=1}^{\Delta+1} \mathcal{F}_{\gamma,\alpha,\beta}$.% be the set of all normal and sub-reverse fans that are processed in the  iteration $\gamma, \alpha$ (all the normal $\alpha\beta$-fans and sub-reverse $\alpha \beta$-fans formed by the algorithm in iteration $\gamma,\alpha$) of the nested for loops in  Line \ref{ln:iteration-gamma-2} and Line \ref{ln:iteration-alpha-2}.% Furthermore, let $M_{\gamma} =\cup_{\alpha=1}^{\Delta+1} M_{\gamma,\alpha}$.
 We again let $V_{\gamma,\alpha}(\phi)$ denote the set of color-$\gamma$ incomplete vertices  that miss color $\alpha$  with respect to the partial coloring $\phi$.

\begin{lemma}
Each iteration of the while loop of Algorithm \ref{alg:fast} colors a constant fraction  of uncolored edges. Thus, it takes $O(\log n)$ iterations to color all the edges.
\end{lemma}
\begin{proof}

For each $uv \in \phi_0^{-1}(\bot)$, we orient the edge from $u$ to $v$ if the color of $u$ is smaller than that of $v$; otherwise we orient it from $v$ to $u$. Specifically, we orient the edge from $u$ to $v$ if $u \in V_\gamma$ and $v \in V_{\gamma'}$ such that $\gamma < \gamma'$. We will show that 

$$\sum_{\alpha =1}^{\Delta+1} |M_{\gamma,\alpha}| +  \sum_{\alpha = 1}^{\Delta+1}|\mathcal{F}_{\gamma,\alpha}| \geq  \sum_{u \in V_\gamma} \outdeg(u) / 8~.$$

For each $v \in V_\gamma$, let $C(v)$ denote the first $\outdeg(v)$ missing colors of $v$ w.r.t.~$\phi_\gamma$ so that $|C(v)| = \outdeg(v)$. For each $\alpha \in C(v)$, if $v \in V_{\gamma,\alpha}(\phi_{\gamma,\alpha})$, then there can be two cases. The first case is that an uncolored edge incident to $v$ will be in the matching $M_{\gamma,\alpha}$. The second case is that  $v$ will belong to a fan in $\mathcal{F}_{\gamma,\alpha}$. Recall that in the second case, the algorithm will first grow a normal $\alpha$-fan from $v$ and if the fan intersects with an existing fan, then $v$ will belong to  a reverse $\alpha$-fan after the merge step.

Otherwise, if $v \notin V_{\gamma,\alpha}(\phi_{\gamma,\alpha})$, it must be caused by either the repairing step of some normal or sub-reverse fan that was grown in iteration $\gamma$. Repairing a normal fan $F$ can only change the missing colors (i.e.~$M(\cdot)$) of its leaves, its center, and the other endpoint of the alternating path. So if $M(v)$ is changed by repairing $F$, we conclude that $v$ is the other end point of the alternating path of $F$ or $v$ is the center of $F$. Note that $v$ cannot be among $F$'s leaves since $V_\gamma$ is an independent set. On the other hand, repairing a sub-reverse fan $F$ can only changes the missing colors of at most five vertices (this corresponds to the case of repairing a sub-reverse fan with three leaves; in fact, a more careful argument shows that a sub-reverse fan can be blamed at most twice). %Observe that coloring an edge in $M_\gamma$ can change $$

Therefore, if $v \notin V_{\gamma,\alpha}(\phi_{\gamma,\alpha})$, we can blame it to the repair of a normal or sub-reverse fan $F$ grown in iteration $\gamma$. We have just argued that each fan $F$ can be blamed at most five times. Furthermore,  a normal or sub-reverse fan in $\mathcal{F}_{\gamma,\alpha}$ has at most one or three vertices in $V_{\gamma,\alpha}(\phi_{\gamma,\alpha})$ respectively. An edge in $M_{\gamma,\alpha}$ has one vertex in $V_{\gamma,\alpha}(\phi_{\gamma,\alpha})$. Hence, we have the following relationship
\begin{align*}
\sum_{u \in V_\gamma} \outdeg(u) =  \sum_{u \in V_\gamma} |C(u)| &\leq \sum_{\alpha=1}^{\Delta+1} |M_{\gamma,\alpha}|+3\sum_{\alpha=1}^{\Delta + 1}|\mathcal{F}_{\gamma,\alpha}| + \mbox{(\#blames)} \\
&\leq \sum_{\alpha=1}^{\Delta+1} |M_{\gamma,\alpha}| + 8 \sum_{\alpha=1}^{\Delta + 1}|\mathcal{F}_{\gamma,\alpha}|~.
 \end{align*}

The second inequality is because $\mbox{(\#blames)} \leq 5\sum_{\alpha=1}^{\Delta + 1}|\mathcal{F}_{\gamma,\alpha}|$. Therefore,
$$\sum_{\alpha =1}^{\Delta+1} |M_{\gamma,\alpha}| +  \sum_{\alpha = 1}^{\Delta+1}|\mathcal{F}_{\gamma,\alpha}| \geq  \sum_{u \in V_\gamma} \outdeg(u) / 8~.$$

Note that at most half of the (normal and sub-reverse) fans are destroyed, since repairing a fan can only destroy at most one other fan.  If a fan is not destroyed, then one uncolored edge in the fan must become colored. Therefore, at least $ \sum_{\gamma=1}^{\Delta+1} \sum_{\alpha=1}^{\Delta+1}|M_{\gamma,\alpha}| + \sum_{\gamma = 1}^{\Delta + 1} \sum_{\alpha = 1}^{\Delta + 1}  |\mathcal{F}_{\gamma,\alpha}|/2$ uncolored edges are colored. The total number of edges that are colored during an iteration of the while loop is at least
$$ \sum_{\gamma=1}^{\Delta+1} \sum_{\alpha=1}^{\Delta+1}|M_{\gamma,\alpha}| + \sum_{\gamma = 1}^{\Delta + 1} \sum_{\alpha = 1}^{\Delta + 1} |\mathcal{F}_{\gamma,\alpha}| / 2  \geq  \sum_{\gamma = 1}^{\Delta + 1} \sum_{v \in V_\gamma} \outdeg(v) / 16 = |\phi_0^{-1}(\bot)| /16~. $$

This implies at least a fraction $1/16$ of the uncolored edges become colored during the iteration. Thus, after $O(\log n)$ iterations, all edges are colored.

\end{proof}

%The running time is dominated by Line  \ref{ln:parallel-agumentation-2a} and Line \ref{ln:parallel-agumentation-2b}, which is executed $O(\Delta^4 \log n)$ times. Each execution uses $O(T \cdot \log n)$ time. The total running time for obtaining a $\Delta + O(\log n / \log \log n)$ coloring is $O(\Delta^4 T \cdot \log^2 n) = O(\Delta^8 \log^3 n)$. We show how to reduce the running time to $O(\Delta^6 \log^3 n)$ in the next section. By Lemma \ref{lem:convert}, it can be converted to a deterministic $(1+\epsilon)\Delta$-edge-coloring algorithm that runs in $O((1/\epsilon)^6 \log^9 n)$ algorithm for $\epsilon = \Omega(\log n / (\Delta \log \log n))$.

It is easy to check that the number of rounds is $O\left(\log n \cdot \Delta^2 \left( \Delta^4 + \Delta \cdot T \cdot \log n\right) \right)$ and the number of parallel augmentation in Line  \ref{ln:parallel-agumentation-2a} and Line \ref{ln:parallel-agumentation-2b} is $t = O(\Delta^3 \log n)$. 

If we opt for a randomized algorithm with $\ell(G)=1$ or a deterministic algorithm with $\ell(G)=O(\log n/(\log \lambda + \log \log n))$, then we set $T = O(\Delta^7 t)$ or $T=2t \lambda$ respectively as in the previous section. Since we managed to reduce $t$ by a factor $\Delta^2$, we obtain the following improvements.

\begin{theorem}
There exists a deterministic distributed $\Delta + O(\log n / (\log \lambda+\log \log n))$-edge-coloring algorithm  that  runs in $O(\Delta^6 \cdot \lambda \log^3 n)$ rounds. Furthermore,
 there exists a deterministic distributed algorithm that colors all the edges using $\Delta + 2$ colors and runs in $O(\Delta^{13} \cdot  \log^3 n)$ rounds.
\end{theorem}

%
%\begin{theorem}
%
%\end{theorem}

In general, we save a factor $\Delta^2$ for both the deterministic and randomized algorithms. One may ask if  we can make further improvement. We next show that it is possible to save another factor $\Delta^2$ if the graph is bipartite.

\subsection{Algorithm for Bipartite Graphs}

\paragraph{Algorithm outline} We now give an even faster and simpler algorithm for edge coloring bipartite graphs. This algorithm  is also based on a similar idea. The pseudo-code is outlined in Algorithm \ref{alg:fast-bipartite}.   

We first fix an iteration of the {\bf while loop}. Then, fix an iteration $\alpha$ of the {\bf for loop} in Line \ref{ln:iteration-alpha-3}. We first find a maximal matching among the uncolored edges $uv$ where $\alpha \in M(u) \cap M(v)$ in  $O(\Delta + \log^{*} n)$ rounds \cite{BEK09}. We color the matching with $\alpha$.  Then, let $V_\alpha$ the independent set of the remaining incomplete vertices that are  still missing $\alpha$.

For bipartite graphs, we however do not need to grow a fan from each vertex in $V_\alpha$. We consider uncolored edges $vu$ for each $v$ in $V_\alpha$.  Some of these edges may share an end point that is not in $V_\alpha$. We can merge these intersecting edges to form reverse $\alpha$-fans. This can easily be implemented in a constant  number of rounds. Now, we are  left with a set of disjoint uncolored \emph{$\alpha$-edges} $Y_\alpha$ that have one end point in $V_\alpha$ and a set of reverse $\alpha$-fans $\mathcal{F_\alpha}$.  All the reverse $\alpha$-fans and $\alpha$-edges are vertex-disjoint. They are also maximal, i.e., all vertices in $V_\alpha$ belong to a reverse $\alpha$-fan or an $\alpha$-edge.

In the normal setting, we can color an edge $vu$ in $Y_\alpha$ as follows. Suppose $\alpha \in M(v)$ and $\beta \in M(u)$. We call $vu$ an \emph{$\alpha \beta$-edge}. First, find the maximal alternating $\alpha \beta$-path $P_u$ from $u$. We observe that because the graph is bipartite, this path cannot end at $u$ or $v$. Hence, we can perform $\augment(P_u)$ and color $vu$ with $\alpha$.

\paragraph{Implementation in the distributed setting} Recall that we currently fix an iteration $\alpha$ in Line \ref{ln:iteration-alpha-3}. We again split reverse fans into sub-reverse  fans. We use \emph{$\alpha \beta$-items} to refer to $\alpha \beta$-edges and sub-reverse $\alpha \beta$-fans.  We aim to repair $\alpha \beta$-items, i.e., color an $\alpha \beta$-edge or repair a sub-reverse $\alpha \beta$-fan in iteration $\beta$ of the for loop in Line \ref{ln:iteration-beta-3}. We want to color $\alpha \beta$-items  simultaneously and efficiently in the \local model. We again truncate an alternating path if its length exceeds $T$ and pick a blocking edge  accordingly  (based on whether the goal is a deterministic algorithm with $\ell(G) = \log n /(\log \lambda + \log \log n)$ or a randomized algorithm with $\ell(G) = 1$). 

A truncated alternating path $P_v(T)$ where $|P_v| > T$ cannot intersect another $\alpha \beta$-item. To this end, we build a conflict graph $G_{\alpha,\beta}$ in which the vertices are $\alpha \beta$-items. There is an edge between two vertices if a truncated alternating path $P$ where $|P| \leq T$ of an $\alpha \beta$-item ends at another $\alpha \beta$-item. With a similar argument as in the previous sections, this graph has bounded arboricity  and therefore can be colored using $O(1)$ colors in $O(T\log n)$ rounds \cite{barenboim2010sublogarithmic}. We then iterate through each color class and repair $\alpha \beta$-edges and  sub-reverse $\alpha \beta$-fans with this color. Since in each color class, the alternating paths neither intersect nor run into another item, we can repair them simultaneously. Similar to the previous algorithm, we also need to repair semi-destroyed sub-reverse $\alpha \beta$-fans.

\begin{algorithm}
\caption{Distributed Item-Repair for Bipartite Graphs}\label{alg:fast-bipartite}
\begin{algorithmic}[1]\small
 \While{there exists uncolored edges}  \label{ln:while-loop-3}
	
	\For{$\alpha =$ color $1,2,\ldots,\Delta$}  \label{ln:iteration-alpha-3}
		\State Let $M_\alpha$ be a maximal matching  of $E_{\alpha} = \{ vu \in E \mid \text{$vu$ is uncolored and }  \alpha \in M(v) \cap M(u) \}$.
		\State Color $M_\alpha$ with $\alpha$. \label{ln:color_maximal-3}
		\State Let $V_{\alpha}$ be the set of incomplete vertices that now miss the color $\alpha$.  
		\State For each $v \in V_\alpha$, pick an incident uncolored edge $vu$. Activate $vu$.
		\State Merge intersecting active edges into reverse $\alpha$-fans.

			 \For{$\beta = $ color $1, 2, \ldots, \Delta + 1$} \label{ln:iteration-beta-3}  \label{ln:maximal-disjoint-3}  
				\State Build a conflict graph $G_{{\alpha \beta}}$ for $\alpha \beta$-edges and sub-reverse $\alpha \beta$-fans.
				\State Color $G_{\alpha \beta}$ using $O(1)$ colors.
				\For{$\eta =$ color $1,2,\ldots O(1)$}
					\State Repair undestroyed $\alpha \beta$-edges and sub-reverse fans in $G_{\alpha \beta}$ with color $\eta$. \label{ln:parallel-agumentation-3a}
				\EndFor
				\State Repair semi-destroyed sub-reverse  fans in $G_{{\alpha \beta}}$. \label{ln:parallel-agumentation-3b}
			\EndFor
	\EndFor

\EndWhile
\end{algorithmic}
\end{algorithm}

Let $V_\alpha(\phi)$ be the set of incomplete vertices  missing $\alpha$ with respect to the partial coloring $\phi$. Let $\phi_0$ be the partial coloring at the beginning of a fixed while loop and $\phi_\alpha$ be the partial coloring right before the iteration $\alpha$ of the loop in Line \ref{ln:iteration-alpha-3}.

\begin{lemma}
Each iteration of the while loop of Algorithm \ref{alg:fast-bipartite} colors a constant fraction  of uncolored edges.  Thus, it takes $O(\log n)$ iterations to color all the edges.
\end{lemma}
\begin{proof}
%Fix an iteration of the while loop. We will show that
%\[
%\sum_{\alpha=1}^{\Delta+1} \left| V_{\alpha}(\phi_{\alpha}) \right| \geq \sum_{u \in V} \outdeg(u)/2~.
%\]
Let $d(v)$ be the number of uncolored edges incident to $v$ with respect to $\phi_0$. Let $C(v)$ be the first  $d(v)$ missing color at $v$ with respect to $\phi_0$. Let   $\mathcal{I}_\alpha$ be the set that  includes all $\alpha \beta$-items and the edges in the matching $M_\alpha$.

Consider $\alpha \in C(v)$. If $v \in V_\alpha(\phi_\alpha)$, then $v$ belongs to an element in  $\mathcal{I}_\alpha$. Specifically,  $v$ may belong to an edge in the matching $M_\alpha$, an $\alpha \beta$-edge or a sub-reverse $\alpha \beta$-fan for some $\beta$. 

If $v \notin V_{\alpha}(\phi_{\alpha})$, then we can blame this to an element in $\mathcal{I}_{\alpha'}$ for $\alpha' < \alpha$. In particular, this could be caused by either a) the repair of an $\alpha'$-item or b) the coloring of some edge in $M_{\alpha'}$. 

The coloring step of each edge in a matching $M_{\alpha'}$ can be blamed at most twice since it can only change the missing colors of its endpoints. The repair of an $\alpha \beta$-edge can be blamed at most three times, since this can only change the missing colors of two of its endpoinds and the other endpoint of its alternating path. Finally, the repair of a sub-reverse $\alpha \beta$-fan can be blamed at most five times. Thus,  each element in $\cup_{\alpha=1}^{\Delta+1} \mathcal{I}_{\alpha}$ can be blamed at most five times overall.

We also note that an edge in $M_\alpha$ has two vertices in $V_\alpha(\phi_\alpha)$. An $\alpha$-edge has exactly one vertex in $V_\alpha(\phi_\alpha)$ and a sub-reverse $\alpha$-fan has at most three vertices in $V_\alpha(\phi_\alpha)$.  Therefore, we have the following relationship

\begin{align*}
\sum_{u \in V} d(u) =  \sum_{u \in V} |C(u)| &\leq 3 \sum_{\alpha=1}^{\Delta }|\mathcal{I}_{\alpha}| + \mbox{(\#blames)} \\
&\leq 8 \sum_{\alpha=1}^{\Delta }|\mathcal{I}_{\gamma,\alpha}| && \mbox{(\#blames)} \leq 5\sum_{\alpha=1}^{\Delta + 1}|\mathcal{I}_{\gamma,\alpha}|~.
 \end{align*}
Using the same argument as in previous sections,  at least half of $\alpha$-items are repaired.  Therefore, the number of  edges being colored in this while iteration is at least 
\begin{align*}
\sum_{\alpha = 1}^{\Delta }|\mathcal{I}_\alpha |/2 & \geq \sum_{v \in V} d(v)/16 = |\phi^{-1}_0(\bot)| /8 ~.
\end{align*}
\end{proof}

The number of parallel augmentations in Line  \ref{ln:parallel-agumentation-3a} and Line \ref{ln:parallel-agumentation-3b} is $t=O(\log n \cdot \Delta^2 )$. Similarly, by setting $T=O(\lambda t)$ or $T = O(t \Delta^7)$, we achieve the following respective improvements. 
\begin{theorem}
Consider bipartite graphs. There exists a deterministic distributed $\Delta + O(\log n / (\log \lambda +\log \log n)$-edge-coloring algorithm  that  runs in $O(\Delta^4 \cdot \lambda \log^3 n)$ rounds. Furthermore, there exists a randomized, $(\Delta+1)$-edge-coloring distributed algorithm that runs in $O(\Delta^{11} \log^3 n)$ rounds.
\end{theorem}
Note that for the randomized case, we use one fewer color in the randomized algorithm since the coloring algorithm requires $m(v)$ to be non-empty only when $v$ is incident to an uncolored edge. Hence, a palette of $\Delta$ colors, excluding the special color $\star$ is sufficient.

\bibliographystyle{alpha}
\bibliography{ref}

\end{document}